\documentclass[a4paper,USenglish,autoref,thm-restate,backref]{lipics-v2021}
\pdfoutput=1
\usepackage{centernot}
\usepackage{shortcuts}
\usepackage[shortcuts]{extdash}

\addto\extrasUSenglish{%
}

\newtheorem{problem}[claim]{Problem}

\title{Fine-Grained Completeness for Optimization in P}
\author{Karl Bringmann}{Saarland University and Max Planck Institute for Informatics,\and Saarland Informatics Campus, Saarbrücken, Germany}{}{}{}
\author{Alejandro Cassis}{Saarland University and Max Planck Institute for Informatics,\and Saarland Informatics Campus, Saarbrücken, Germany}{}{}{}
\author{Nick Fischer}{Saarland University and Max Planck Institute for Informatics,\and Saarland Informatics Campus, Saarbrücken, Germany}{}{}{}
\author{Marvin Künnemann}{Institute for Theoretical Studies, ETH Z\"urich, Switzerland}{}{}{}

\authorrunning{K. Bringmann, A. Cassis, N. Fischer, and M. Künnemann}
\Copyright{Karl Bringmann, Alejandro Cassis, Nick Fischer, and Marvin Künnemann}

\funding{\emph{Karl Bringmann, Alejandro Cassis, Nick Fischer:} This work is part of the project TIPEA that has received funding from the European Research Council (ERC) under the European Unions Horizon 2020 research and innovation programme (grant agreement No.~850979). \emph{Marvin Künnemann:} Research supported by Dr.\ Max Rössler, by the Walter Haefner Foundation, and by the ETH Zürich Foundation. Part of this research was performed while the author was employed at Max Planck Institute for Informatics.}

\nolinenumbers
\hideLIPIcs

\begin{document}
\maketitle

\begin{abstract}
We initiate the study of fine-grained completeness theorems for exact and approximate optimization in the polynomial-time regime.

Inspired by the first completeness results for decision problems in P (Gao, Impagliazzo, Kolokolova, Williams, TALG 2019) as well as the classic class \MaxSNP{} and \MaxSNP{}-completeness for \NP{} optimization problems (Papadimitriou, Yannakakis, JCSS 1991), we define polynomial-time analogues \MaxSP{} and \MinSP{}, which contain a number of natural optimization problems in P, including Maximum Inner Product, general forms of nearest neighbor search and optimization variants of the $k$-XOR problem. Specifically, we define \MaxSP{} as the class of problems definable as $\max_{x_1,\dots,x_k} \#\{ (y_1,\dots,y_\ell) : \phi(x_1,\dots,x_k, y_1,\dots,y_\ell) \}$, where $\phi$ is a quantifier-free first-order property over a given relational structure (with \MinSP{} defined analogously). On $m$-sized structures, we can solve each such problem in time $O(m^{k+\ell-1})$. 
Our results are:
\begin{itemize}
	\item We determine (a sparse variant of) the Maximum/Minimum Inner Product problem as complete under \emph{deterministic} fine-grained reductions: A strongly subquadratic algorithm for Maximum/Minimum Inner Product would beat the baseline running time of $O(m^{k+\ell-1})$ for \emph{all} problems in \MaxSP{}/\MinSP{} by a polynomial factor.
	\item This completeness transfers to approximation: Maximum/Minimum Inner Product is also complete in the sense that a strongly subquadratic \emph{$c$-approximation} would give a \emph{$(c+\varepsilon)$-approximation} for all \MaxSP{}/\MinSP{} problems in time $O(m^{k+\ell-1-\delta})$, where $\varepsilon > 0$ can be chosen arbitrarily small. Combining our completeness with~(Chen, Williams, SODA 2019), we obtain the perhaps surprising consequence that refuting the \OV{} Hypothesis is \emph{equivalent} to giving a $O(1)$-approximation for all \MinSP{} problems in faster-than-$O(m^{k+\ell-1})$ time.
	\item By fine-tuning our reductions, we obtain mild algorithmic improvements for solving and approximating all problems in \MaxSP{} and \MinSP{}, using the fastest known algorithms for Maximum/Minimum Inner Product.
\end{itemize}
\end{abstract}

\section{Introduction}
For decades, increasingly strong hardness of approximation techniques have been developed to pinpoint the best approximation guarantees achievable in polynomial time. Among the early successes of the field, we find the \MaxSNP{} completeness theorems by Papadimitriou and Yannakakis~\cite{PapadimitriouY91}, giving the first strong evidence against PTASes for \MaxSAT{} and related problems. Such completeness theorems constitute valuable tools in complexity theory: Generally speaking, proving a problem $A$ to be complete for a class $\class$ shows that $A$ is \emph{the} representing problem for~$\class$. The precise notion of completeness is typically chosen such that a certain algorithm for $A$ would yield unexpected algorithms for the whole class $\class$ -- thus establishing that $A$ is unlikely to admit such an algorithm. However, a completeness result may also open up algorithmic uses. Namely, since any problem in $\class$ can be reduced to its complete problem~$A$,  we may find (possibly mildly) improved algorithms for all problems in~$\class$ by making algorithmic progress on the single problem $A$.

Given this usefulness, it may be surprising that there are currently no completeness results for studying optimization barriers \emph{within the polynomial-time regime}, e.g., for approximability in strongly subquadratic time (in fact, even for studying decision problems, completeness results are an exception rather than the norm, see~\cite{VassilevskaW18} for a recent survey of the field). Thus, this work sets out to initiate the quest for completeness results for optimization in P, which corresponds to studying the (in-)approximability of problems on large data sets.

\paragraph*{Previous Completeness Results in P}
The essentially only known completeness result in fine-grained complexity theory in P is a recent result by Gao, Impagliazzo, Kolokolova, and Williams~\cite{GaoIKW18}: The orthogonal vectors problem (OV)\footnote{Given two sets of $n$ vectors in $\{0,1\}^d$, determine whether there exists a pair of vectors, one of each set, that are orthogonal.} is established as complete problem for the class of model-checking first-order properties\footnote{Let $\phi$ be a first-order property (in prenex normal form) over a relational structure of size $m$. Given the structure, determine whether $\phi$ holds. See Section~\ref{sec:preliminaries} for details.} under fine-grained reductions\footnote{For a formal definition of fine-grained reductions, see~\cite{CarmosinoGIMPS16, GaoIKW18}. For this paper, the reader may think of the following slightly simpler notion: A fine-grained reduction from a problem $P_1$ with presumed time complexity $T_1$ to a problem $P_2$ with presumed time complexity $T_2$ is an algorithm $A$ for $P_1$ that has oracle access to $P_2$ and whenever we use an $O(T_2(n)^{1-\delta})$ algorithm for the calls to the $P_2$-oracle (for some $\delta > 0$), there is a $\delta' > 0$ such that $A$ runs in time $O(T_1(n)^{1-\delta'})$.}. From this completeness, they derive in particular:
\begin{itemize}
\item \textbf{\textsf{Hardness:}} If there are $\gamma, \delta >0$ such that OV with moderate dimension $d=n^\gamma$ can be solved in time $O(n^{2-\delta})$, then there is some $\delta' > 0$ such that all $(k+1)$-quantifier first-order properties can be model-checked in time $O(m^{k-\delta'})$ for $k \geq 2$. The negation of this statement's premise is known as the moderate-dimensional OV Hypothesis; the consequence would be very surprising, as model-checking first-order properties is a very general class of problems for which no $O(m^{k-\delta})$-time algorithm is known. This result can be seen as support for the moderate-dimensional OV Hypothesis. 
\item \textbf{\textsf{Algorithms:}} Using a stronger notion than fine-grained reductions, Gao et al.\ also prove that mildly subquadratic algorithms for OV have algorithmic consequences for model-checking first-order properties. Specifically, by combining their reductions with the fastest known algorithm for OV~\cite{AbboudWY15, ChanW16}, they obtain an $m^{k}/2^{\Omega(\sqrt{\log m})}$-time algorithm for model-checking any $(k+1)$-quantifier first-order property. 
\end{itemize}
No comparable fine-grained completeness results are known for polynomial-time optimization problems, raising the question:
Can we give completeness theorems also for a general class of optimization problems in P, both for exact and approximate computation?

\paragraph*{Hardness of Approximation in P}
Studying the fine-grained approximability of polynomial-time optimization problems (hardness of approximation in P), is a recent and influential trend: After a breakthrough result by Abboud, Rubinstein, and Williams~\cite{AbboudRW17} establishing the \emph{Distributed PCP in P} framework, a number of works gave strong conditional lower bounds, including results for nearest neighbor search~\cite{Rubinstein18} or a tight characterization of the approximability of maximum inner product~\cite{Chen18, ChenW19}. Further results include work on approximating graph problems~\cite{RodittyVW13,BackursRSVWW18, BringmannKW19, KarthikLM19}, the Fréchet distance~\cite{Bringmann14}, LCS~\cite{AbboudB17, AbboudR18}, monochromatic inner product~\cite{KarthikM19}, earth mover distance~\cite{Rohatgi19}, as well as equivalences for fine-grained approximation in P~\cite{ChenW19, ChenGLRR19, BringmannKW19}. Related work studies the inapproximability of parameterized problems, ruling out certain approximation guarantees within running time $f(k) n^{g(k)}$ under parameter $k$ (such as FPT time $f(k)\poly(n)$, or $n^{o(k)}$), see~\cite{FeldmannKLM20} for a recent survey.\footnote{Note that these parameterized inapproximability results do not necessarily apply to the case of a fixed parameter~$k$, which would correspond to our setting. See~\cite{KarthikLM19} for an interesting exception.}
	 
\paragraph*{An Optimization Class: Polynomial-Time Analogues of MaxSNP}
We define a natural and interesting class of polynomial-time optimization problems, inspired by the approach of Gao et al.~\cite{GaoIKW18} as well as the classic class \MaxSNP{} introduced by Papadimitriou and Yannakakis~\cite{PapadimitriouY91} to study the approximability of NP optimization problems.

The definition of \MaxSNP{} is motivated by Fagin's theorem (see, e.g., \cite{Immerman99, GraedelKLMSVVW07}), which characterizes $\NP$ as the family of problems expressible as $\exists S\, \forall\bar{y}\, \exists\bar{z}\, \phi(\bar{y},\bar{z},G,S)$ where $G$ is a given relational structure, $\exists S$ ranges over a relational structure $S$ and $\forall\bar{y}\, \exists\bar{z}\, \phi(\bar{y},\bar{z},G,S)$ is a $\forall^* \exists^*$-quantified first-order property. A subclass of this is $\SNP$, which consists of those problems expressible without the $\exists\bar{z}$-part. Its natural optimization variant is $\MaxSNP$, defined as the set of problems expressible as $\max_{S}\#\{\bar{y} : \phi(\bar{y},G,S)\}$. Notably, this class of problems contains central optimization problems (\MAXThreeSAT{}, \MaxCut{}, etc.), all of which admit a constant-factor approximation in polynomial time. Using a notion of $\MaxSNP$-completeness, Papadimitriou and Yannakakis identified several problems (including \MAXThreeSAT{} and \MaxCut{}) as hardest-to-approximate in this class, giving a justification for the lack of a PTAS for these problems.\footnote{A stronger justification was later given by the PCP theorem, establishing inapproximability even under $\P \ne \NP$. In general, these two approaches (approximation-preserving completeness theorems as well as proving inapproximability under established assumptions on exact computation) can result in incomparable hardness of approximation results.}

To study the same type of questions in the polynomial-time regime, the perhaps most natural approach is to restrict the syntax defining \MaxSNP{} problems such that it solely contains polynomial-time problems. Specifically, we replace $\max_S$ by a maximization over a bounded number of $k$ variables $x_1,\dots,x_k$ and restrict the counting operator to tuples $\bar{y} = (y_1,\dots,y_\ell)$ of bounded length $\ell$. The resulting formula $\max_{x_1,\dots,x_k} \#\{(y_1,\dots,y_\ell) : \phi(x_1,\dots,x_k,y_1,\dots,y_\ell)\}$ can be easily seen (see Appendix~\ref{sec:baseline}) to be solvable in time $O(m^{k+\ell-1})$, where $m$ denotes the problem size. We define $\MaxSP_{k,\ell}$ to denote the class of these optimization problems and let $\MaxSP = \bigcup_{k\ge 2, \ell \ge 1} \MaxSP_{k,\ell}$. Note that here, ``$\mathrm{SP}$'' stands for ``strict P'' in analogy to the name ``strict NP'' of $\SNP$. We refer to Section~\ref{sec:preliminaries} for more details.

We obtain an analogous minimization class \MinSP{} by replacing $\max$ by $\min$ everywhere. These classes include interesting problems:

\begin{itemize}
\item Vector-definable problems: Let $\Sigma = \{0,\dots, c\}$ be a fixed alphabet and $f: \Sigma^k \to \{0,1\}$ be an arbitrary Boolean function. Then we can express the following problem: Given sets $X_1,\dots, X_k$ in $\Sigma^d$ of vectors, maximize (or minimize) $\sum_{i=1}^d f(x_1[i],\dots,x_k[i])$ over all $x_1\in X_1, \dots, x_k\in X_k$. Each such problem is definable in $\MaxSP_{k,1}$/$\MinSP_{k,1}$, e.g.:
\begin{itemize}
\item Maximum Inner Product (\MaxIP{}): Given sets $X_1,X_2 \subseteq \{0,1\}^d$, maximize the inner product $x_1 \cdot x_2$ over $x_1\in X_1,x_2\in X_2$. To see that this problem is in $\MaxSP_{2,1}$, consider the formula  $\max_{x_1\in X_1,x_2\in X_2} \#\{y \in Y : E(x_1,y) \wedge E(x_2,y)\}$, where $E(x,y)$ indicates that the $y$-th coordinate of $x$ is equal to $1$.\item Consider minimization with $k=2$ and view $f: \Sigma^2 \to \{0,1\}$ as classifying pairs of characters as \emph{similar} (0) or \emph{dissimilar} (1). This expresses the following problem that generalizes the nearest-neighbor problem over the Hamming metric: Given a set of length-$d$ strings over $\Sigma$, determine the most similar pair of strings by minimizing the number of dissimilar characters.
\item View $\Sigma$ as the finite field $\mathbb{F}_q$ and let $f(z_1,\dots,z_k) = 1$ iff $\sum_{i=1}^k z_i \equiv 0\pmod q$. This gives optimization variants of the \kXOR{k} problem~\cite{JafargholiV16,DietzfelbingerSW18}, generalized to arbitrary finite fields. 
\end{itemize}
\item Beyond vector-definable problems, in $\MaxSP_{2,\ell-2}$ we can express the graph problem of computing, over all edges $e$, the maximum number of length-$\ell$ circuits containing $e$: \[\max_{x_1,x_2} \#\{ (y_1,\dots, y_{\ell-2}) : E(x_1,x_2) \wedge E(x_2,y_1) \wedge \cdots \wedge E(y_{\ell-3},y_{\ell-2}) \wedge E(y_{\ell-2}, x_1)\}.\]
	In fact, $\MaxSP$ also contains generalizations of this problem to other pattern graphs than length-$\ell$ circuits (e.g., length-$\ell$ cycles or $\ell$-cliques), even arbitrary fixed patterns in \emph{hypergraphs}.
\end{itemize}
We let $m$ denote the size of the relational structure, that is, the number of tuples in an explicit representation of all relations. For vector-definable examples, the input can be represented as a relational structure of size $m=O(n d\log |\Sigma|)$, which is the natural input size. Note, however, that the relational structure also allows us to succinctly encode \emph{sparse} vectors in very large dimension (such as $d=\Theta(n)$), which is why we often refer to \MaxSP{} and \MinSP{} as describing a \emph{sparse setting}.
It is easy to see that each \MaxSP{} or \MinSP{} formula $\psi$ can be solved in time $O(m^{k+\ell-1})$ (see Appendix~\ref{sec:baseline}); note that for a fixed $\psi$, $k$ and $\ell$ always denote the number of maximization/minimization and counting variables, respectively. Can we obtain completeness results with respect to improvements over this baseline running time?

\paragraph*{(Sparse) Maximum Inner Product}
Our results prove the Maximum Inner Product problem (\MaxIP{}) as representative for the class \MaxSP{}. We will formally introduce two important variants of this problem.

\begin{problem}[\MaxIP{}]
Given two sets of $n$ vectors $X_1, X_2 \subseteq \{0, 1\}^d$, the task is to compute the maximum inner product $\langle x_1, x_2\rangle = \sum_j x_1[j] \mult x_2[j]$ for $x_1 \in X_1, x_2 \in X_2$.
\end{problem}

When $d = n^\gamma$ for some (small) $\gamma > 0$, we speak of the \emph{moderate-dimensional} \MaxIP{} problem. In this paper, we also use \MaxIP{} in another context, depending on the input format. To make the distinction explicit, let us formally introduce the \emph{Sparse} Maximum Inner Product problem (\SparseMaxIP{}):

\begin{problem}[\SparseMaxIP{}]
Given two sets of $n$ vectors $X_1, X_2 \subseteq \{0, 1\}^d$, sparsely represented as a list of pairs $(x_i, j)$ which represent the one-coordinates $x_i[j] = 1$, the task is to compute the maximum inner product $\langle x_1, x_2\rangle$ for $x_1 \in X_1, x_2 \in X_2$.
\end{problem}

For moderate-dimensional \MaxIP{} we measure the complexity in $n$ and for \SparseMaxIP{} we measure the complexity in $m$, the total number of one-coordinates. We note that \SparseMaxIP{} is also special in our setting as this problem can be seen as a member of $\MaxSP_{2,1}$. Indeed, \SparseMaxIP{} is the same problem as maximizing the formula
\begin{equation*}
    \psi = \max_{x_1 \in X_1, x_2 \in X_2} \#\{y \in [d] : E(x_1, y) \land E(x_2, y) \}, 
\end{equation*}
where $E(x_i, y)$ indicates that the $y$-th coordinate of $x_i$ is equal to $1$. We also define the (Sparse) Minimum Inner Product problems (\MinIP{}, \SparseMinIP{}) as the analogous problems with the task to minimize $\innerprod{x_1}{x_2}$.

\subsection{Our Results}
Our first main result is a completeness theorem for exact optimization, establishing Maximum Inner Product as complete for \MaxSP{} (and Minimum Inner Product for \MinSP{}).

\begin{theorem}[\SparseMaxIP{} is \MaxSP{}-complete] \label{thm:completeness-ip}
\SparseMaxIP{} is complete for the class \MaxSP{} under fine-grained reductions: If there is some $\delta > 0$ such that \SparseMaxIP{} can be solved in time $O(m^{2-\delta})$, then for every $\MaxSP_{k,\ell}$ formula $\psi$, there is some $\delta'>0$ such that $\psi$ can be solved in time $O(m^{k+\ell-1-\delta'})$.

The analogous statement holds for minimization, if we replace \SparseMaxIP{} and \MaxSP{} by \SparseMinIP{} and \MinSP{}, respectively.
\end{theorem}

Turning to the approximability of \MaxSP{} and \MinSP{}, we show how to obtain a fine-grained completeness that even preserves approximation factors (up to an arbitrarily small blow-up). Here and throughout the paper, we say that an algorithm gives a $c$-approximation for a maximization problem if it outputs a value in the interval $[c^{-1}\cdot\OPT, \OPT]$, where $\OPT$ is the optimal value. For minimization, the algorithm computes a value in the interval $[\OPT, c\cdot \OPT]$. 

\begin{theorem}[\SparseMaxIP{} is \MaxSP{}-complete, (almost) approximation preserving] \label{thm:completeness-ip-approx}
Let $c\ge 1$ and $\varepsilon > 0$. If there is some $\delta > 0$ such that \SparseMaxIP{} can be $c$-approximated in  time $O(m^{2-\delta})$, then for every $\MaxSP_{k,\ell}$ formula $\psi$, there is some $\delta'>0$ such that $\psi$ can be $(c+\varepsilon)$-approximated in time $O(m^{k+\ell-1-\delta'})$.

The analogous statement for minimization holds for \SparseMinIP{} and \MinSP{}. 
\end{theorem}

As a key technical step to obtain Theorems~\ref{thm:completeness-ip} and~\ref{thm:completeness-ip-approx}, we prove a \emph{universe reduction} for \MaxSP{}/\MinSP{}
formulas (detailed in Sections~\ref{sec:technical-overview} and~\ref{sec:universe-reduction}). Along the way, this universe reduction establishes the following fine-grained equivalence between the sparse and moderate-dimensional settings of \MaxIP{}/\MinIP{}. 

\begin{theorem}[Equivalence between \MaxIP{} and \SparseMaxIP{}]\label{thm:equivalence-ip}
\hspace{0cm}
\begin{itemize}
\item There are some $\gamma, \delta > 0$ such that \MaxIP{} with dimension $d=n^\gamma$ can be solved in time $\Order(n^{2-\delta})$ if and only if there is some $\delta' > 0$ such that \SparseMaxIP{} can be solved in time $\Order(m^{2-\delta'})$.
\item Let $c> 1$ and $\varepsilon > 0$. If there are some $\gamma, \delta > 0$ such that \MaxIP{} with dimension $d=n^\gamma$ can be $c$-approximated in  time $\Order(n^{2-\delta})$ then there is some $\delta' > 0$ such that \SparseMaxIP{} can be $(c+\varepsilon)$-approximated in time $\Order(m^{2-\delta'})$. Conversely, if there is some $\delta > 0$ such that \SparseMaxIP{} can be $c$-approximated in time $O(m^{2-\delta})$ then there 
are some $\gamma, \delta' > 0$ such that \MaxIP{} with dimension $d = n^{\gamma}$ can be $c$-approximated in time $O(n^{2-\delta'})$.
\end{itemize}
The analogous statements for minimization hold for \MinIP{}.
\end{theorem}

We prove Theorems~\ref{thm:completeness-ip},~\ref{thm:completeness-ip-approx} and~\ref{thm:equivalence-ip} in \autoref{sec:consequences}.

\paragraph*{Consequences for Hardness of Approximation}
As a consequence of the above completeness results and dimension reduction, we obtain the following statements.
\begin{itemize}
\item Since Maximum Inner Product and Minimum Inner Product are subquadratic equivalent in moderate dimensions~\cite[Theorem 1.6]{ChenW19}, we obtain from Theorems~\ref{thm:completeness-ip} and~\ref{thm:equivalence-ip} that a strongly subquadratic algorithm solving moderate-dimensional Maximum Inner Product \emph{exactly} would give a polynomial-factor improvement over the $O(m^{k+\ell-1})$ running time for all \MaxSP{} \emph{and} \MinSP{} formulas. This adds an additional surprising consequence of fast Maximum Inner Product algorithms, besides refuting the Orthogonal Vectors Hypothesis.
\item There is a $O(1)$-approximation beating the quadratic baseline for moderate-dimensional Maximum Inner Product \emph{if and only if} there is a $O(1)$-approximation beating the $O(m^{k+\ell-1})$ time baseline for all \MaxSP{} formulas. To obtain this result combine the fine-grained equivalence of $O(1)$-approximation of moderate-dimensional \MaxIP{} and \SparseMaxIP{} (\autoref{thm:equivalence-ip}) with the completeness of \SparseMaxIP{} (\autoref{thm:completeness-ip-approx}). This adds an additional consequence of fast Maximum Inner Product approximation, besides refuting SETH~\cite{AbboudRW17,Chen18}.
\item In the minimization world, we obtain a tight connection between approximating \MinSP{} formulas and \OV{}: The (moderate-dimensional) \OV{} hypothesis is \emph{equivalent} to the non-existence of a $O(1)$-approximation for all \MinSP{} formulas in time $O(m^{k+\ell-1})$. To obtain this result, combine the equivalence of moderate-dimensional \OV{} Hypothesis and non-existence of a $O(1)$-approximation for moderate-dimensional \MinIP{}~\cite[Theorem~1.5]{ChenW19} with the equivalence of $O(1)$-approximation algorithms for moderate-dimensional \MinIP{} and \MinSP{} (Theorem~\ref{thm:completeness-ip-approx} and \autoref{thm:equivalence-ip}). Interestingly, this can be seen as additional support for the Orthogonal Vectors Hypothesis. 
\end{itemize}

\paragraph*{Algorithms: Lower-Order Improvements}
Since Maximum Inner Product has received significant interest for improved algorithms (see particularly~\cite{Chen18, ChenW19}), we turn to the question whether our completeness result also yields  lower-order algorithmic improvements
for all problems in the class. Indeed, by combining the best known Maximum/Minimum Inner Product algorithms with our reductions, we obtain the following general results for \MaxSP{} and \MinSP{}. We give the proofs for both theorems in~\autoref{sec:consequences}.

\begin{theorem}[Lower-Order Improvement for Exact \MaxSP{} and \MinSP{}] \label{thm:lower-order-exact}
We can exactly optimize any $\MaxSP_{k,\ell}$ and $\MinSP_{k,\ell}$ formula in randomized time $m^{k+\ell - 1}/\log^{\Omega(1)} m$.
\end{theorem}

Interestingly, for constant-factor approximations, a complete shave of logarithmic factors is possible.

\begin{theorem}[Lower-Order Improvement for Approximate \MaxSP{} and \MinSP{}] \label{thm:lower-order-approx}
For every constant $c>1$, we can $c$-approximate every $\MaxSP_{k,\ell}$ and $\MinSP_{k,\ell}$ formula in time \raisebox{0pt}[0pt][0pt]{$m^{k+\ell - 1}/2^{\Omega(\sqrt{\log m})}$}. For $\MaxSP_{k,\ell}$ the algorithm is deterministic; for $\MinSP_{k,\ell}$ it uses randomization.
\end{theorem}

\section{Preliminaries} \label{sec:preliminaries}
For an integer $k \geq 1$, we set $[k] = \{1, \ldots, k\}$. Moreover, we write $\widetilde O(T) = T \log^{\Order(1)} T$.

\paragraph*{First-Order Model-Checking}
A \emph{relational structure $(X, R_1, \dots, R_r)$} consists of $n$ objects $X$ and relations $R_j \subseteq X^{a_j}$ (of arbitrary arities~$a_j$) between these objects. A \emph{first-order formula} is a quantified formula of the form
\begin{equation*}
  \psi = (Q_1 x_1)\, \ldots\, (Q_k x_k)\, \phi(x_1, \ldots, x_k),
\end{equation*}
where $Q_i \in \{\exists, \forall\}$ and $\phi$ is a Boolean formula over the predicates $R_j(x_{i_1}, \dots, x_{i_{a_j}})$. Given a relational structure, the \emph{model-checking problem} (or \emph{query evaluation problem}) is to check whether~$\psi$ holds on the given structure, that is, for $x_1, \dots, x_k$ ranging over $X$ and by instantiating the predicates $R_j(x_{i_1}, \dots, x_{i_{a_j}})$ in $\phi$ according to the structure, $\psi$ is valid.

Following previous work in this line of research~\cite{GaoIKW18,BringmannFK19}, we assume that the input is represented sparsely -- that is, we assume that the relational structure is written down as an exhaustive enumeration of all records in all relations; let $m$ denote the total number of such entries. This convention is reasonable as this data format is common in the context of database theory and also for the representation of graphs (where it is called the \emph{adjacency list} representation). By ignoring objects not occurring in any relation, we may always assume that $n \leq \Order(m)$.

It is often convenient to assume that each variable $x_i$ ranges over a separate set $X_i$. We can make this assumption without loss generality, by introducing some additional unary predicates.

\paragraph*{\texorpdfstring{\boldmath$\MaxSP_{k,\ell}$ and $\MinSP_{k,\ell}$}{MaxSPk and MinSPk}}
In analogy to first-order properties with quantifier structure $\exists^k\forall^{\ell}$ (with maximization instead of $\exists$ and counting instead of $\forall$), we now define a class of optimization problems: Let $\MaxSP_{k,\ell}$ be the class containing all formulas of the form 
\begin{equation} \label{eq:psi}
  \psi = \max_{x_1, \dots, x_k} \counting_{y_1, \dots, y_{\ell}}\, \phi(x_1, \ldots, x_k, y_1, \dots, y_{\ell}),
\end{equation}
where, as before, $\phi$ is a Boolean formula over some predicates of arbitrary arities. We similarly define $\MinSP_{k,\ell}$ with ``$\min$'' in place of ``$\max$''. Occasionally, we write $\OptSP_{k,\ell}$ to refer to both of these classes simultaneously, and we write ``$\opt$'' as a placeholder for either ``$\max$'' or ``$\min$''. In analogy to the model-checking problem for first-order properties, we associate to each formula $\psi \in \OptSP_{k,\ell}$ an algorithmic problem:

\begin{definition}[$\Max(\psi)$ and $\Min(\psi)$] \label{def:optimization}
Let $\psi \in \MaxSP_{k,\ell}$ be as in~\eqref{eq:psi}. Given a relational structure on objects $X$, the $\Max(\psi)$ problem is to compute
\begin{equation*}
  \OPT = \max_{x_1, \dots, x_k \in X} \,\counting_{y_1, \dots, y_{\ell} \in X}\, \phi(x_1, \dots, x_k, y_1, \dots, y_{\ell}).
\end{equation*}
We similarly define $\Min(\psi)$ for $\psi \in \MinSP_{k,\ell}$. Occasionally, for $\psi \in \OptSP_{k,\ell}$, we write $\Opt(\psi)$ to refer to both problems simultaneously.
\end{definition}

As before, we usually assume (without loss of generality) that each variable ranges over a separate set: $x_i \in X_i$, $y_i \in Y_i$. In particular, as claimed before we can express the \SparseMaxIP{} formula
\begin{equation*}
  \psi = \max_{x_1 \in X_1, x_2 \in X_2} \#\{y \in [d] : E(x_1, y) \land E(x_2, y) \}
\end{equation*}
in a way which is consistent with \autoref{def:optimization} by introducing three unary predicates for $X_1$, $X_2$ and $[d]$. For convenience, we introduce some further notation: For objects $x_1 \in X_1, \dots, x_k \in X_k$, we denote by $\Val(x_1, \dots, x_k) = \counting_{y_1, \dots, y_{\ell}} \phi(x_1, \dots, x_k, y_1, \dots, y_{\ell})$ the \emph{value} of $(x_1, \dots, x_k)$.

\autoref{def:optimization} introduces $\Max(\psi)$ and $\Min(\psi)$ as \emph{exact} optimization problems (i.e., $\OPT$ is required to be computed exactly). We say that an algorithm computes a \emph{$c$\=/approximation} for $\Max(\psi)$ if it computes any value in the interval $[c^{-1} \mult \OPT, \OPT]$. Similarly, a $c$\=/approximation for $\Min(\psi)$ computes any value in $[\OPT, c \mult \OPT]$.

The problem $\Opt(\psi)$ can be solved in time $\Order(m^{k+\ell-1})$ for all $\OptSP_{k,\ell}$ formulas $\psi$, by a straightforward extension of the model-checking baseline algorithm; see~\autoref{sec:baseline} for details. As this is clearly optimal for $k + \ell = 2$, we will often implicitly assume that $k + \ell \geq 3$ in the following.

As we show in~\autoref{sec:improved-alg-multiple-counting}, we can \emph{exactly} solve $\OptSP_{k,\ell}$ in time $O(m^{k + \ell - 3/2})$ when $\ell \geq 2$. Thus, in the remaining sections we will be working with the \emph{hardest} case $\ell = 1$. For convenience we write $\MaxSP_k := \MaxSP_{k,1}$,  and similarly for $\MinSP_k$ and $\OptSP_k$. Since for a fixed formula $\psi\in \OptSP$, $k$ and $\ell$ are constants, $f(k,\ell)$-factors are hidden in the $O$-notation throughout the paper.
\section{Technical Overview}\label{sec:technical-overview}

In this section we give an overview of the main technical ideas used to give our completeness result (\autoref{thm:completeness-ip}). Let $\psi$ be a $\MaxSP{}_{k, \ell}$ formula. We will outline the reduction from $\Max(\psi)$ to \SparseMaxIP{}. Since for $\ell \geq 2$ we can solve $\Max(\psi)$ in time $O(m^{k+\ell-3/2})$ (see~\autoref{sec:improved-alg-multiple-counting}) we focus on the case of $\ell = 1$. The reduction consists of two phases. In the first phase (\autoref{sec:reduction-to-hybrid}), we reduce $\psi$ to an intermediate problem called the \emph{Hybrid Problem} which captures the core hardness, but is more restricted. 
For now, the reader can think of the Hybrid Problem as a vector-definable problem (as introduced in the introduction) $\max_{x_1\in X_1,\dots,x_k\in X_k} \sum_{i=1}^{d} f(x_1[i],...,x_k[i])$ with $X_1,\dots,X_k \subseteq \{0,1\}^d$; we define it formally in~\autoref{sec:hybrid-problem}. Since a Hybrid Problem is more restricted than the general problem $\Max(\psi)$, the first phase consists of the following 4 steps in which we progressively restrict the shape of $\psi$:

\begin{enumerate}
\item Remove all \emph{hyperedges}, that is, $\psi$ no longer contains predicates of arity $\geq 3$ so an instance of $\Max(\psi)$ can be thought of as a graph with parallel (or alternatively, colored) edges. 
\item Remove all edges between vertices $x_i$ and $x_j$ that we maximize over. We will call these \emph{cross edges}. After this step the only remaining edges are between vertices $x_i$ and the counting variable $y$.
\item Remove all \emph{parallel edges} (or alternatively, colored edges), that is, we combine parallel edges into simple edges.
\item Remove unary predicates, finally turning the $\Max(\psi)$ instances into graphs. At this point it becomes simple to rewrite $\Max(\psi)$ as a Hybrid Problem.
\end{enumerate}

The second phase of the reduction is to reduce the Hybrid Problem to a \SparseMaxIP{} instance~(\autoref{sec:universe-reduction}). The general idea of this step seems straightforward: For simplicity again let us focus on a vector-definable problem $\max_{x_1\in X_1,\dots,x_k\in X_k} \sum_{i=1}^{d} f(x_1[i],...,x_k[i])$ with $X_1,\dots,X_k \subseteq \{0,1\}^d$. We can precisely ``cover'' each $f(x_1[i],\dots,x_k[i])$ by at most $2^k$ summands expressing
\begin{equation*}
  \sum_{\substack{\alpha_1,\dots,\alpha_k \in \{0,1\} \\ f(\alpha_1,\dots,\alpha_k)=1}} [(x_1[i],\dots,x_k[i])=(\alpha_1,\dots,\alpha_k)],  
\end{equation*}
where the outer $[ \cdot ]$ denotes the Iverson brackets. Observe that each such summand is equivalent to the \MaxIP{} function, \emph{up to complementing some $x_j[i]$'s} (i.e.\ each summand can be expressed as \MaxIP{} by setting $x_j[i] := 1 - x_j[i]$ whenever $\alpha_j = 0$). The issue, however, is that complementing $x_j[i]$'s means complementing a binary relation of size $O(m)$ (between~$n$ vectors and $d$ coordinates). Since complementing a sparse relation generally produces a dense relation (here: of size $\Omega(nd)$), this will produce a prohibitively large problem size for the \SparseMaxIP{} formulation if $d$ is large. 

The natural approach to overcome this issue is to reduce the dimension of the Hybrid Problem, so that we can afford the complementation step. One challenge in this is that \MaxSP{} formula might have its optimal objective value anywhere in $\{0,\dots,m^\ell\}$, but reducing the dimension from $d\le m^\ell$ to, say, $d=m^\gamma$ also reduces the range of possible objective values to $\{0,\dots,m^\gamma\}$. It appears counter-intuitive that such a ``compression'' of objective values should be possible while allowing us to reconstruct the optimum value \emph{exactly}. Perhaps surprisingly, 
we are able to achieve this by a simple deterministic dimension reduction.

The idea of our dimension reduction is as follows. For concreteness, focus on the \SparseMaxIP{} problem. Starting from a \SparseMaxIP{} instance $X_1, X_2 \subseteq \{0, 1\}^d$, we construct a hash function $h: \{0,1\}^d \mapsto \{0,1\}^{d'}$ with $d' \ll d$, which maps every one-entry to $t$ coordinates in $[d']$. More precisely, for every coordinate $i \in [d]$, we deterministically choose an auxiliary vector $w_i \in \{0, 1\}^{d'}$ with exactly $t$ one-entries for some parameter~$t$. Then, the hash function is defined as $h(x) = \bigvee_{i : x[i] = 1} w_i$ (here the OR is applied coordinate-wise).

We say that there is a collision between two vectors $x_1, x_2$ if there are distinct $i, j \in [d]$ such that $x_1[i] = x_2[j] = 1$ and the auxiliary vectors $w_i$ and $w_j$ share a common one-entry. Ideally, every pair of vectors $x_1 \in X_1,x_2 \in X_2$ is hashed \emph{perfectly}, meaning that no collision takes place. In that case, it holds that $\innerprod{h(x_1)}{h(x_2)} = t \cdot \innerprod{x_1}{x_2}$ and thus also $\OPT' = t \mult \OPT$, where $\OPT$ and $\OPT'$ are the objective values of the original and the hashed instance, respectively. However, in reality we cannot expect the hashing to be perfect. Note that nevertheless the difference $|\innerprod{h(x_1)}{h(x_2)} - t \cdot \innerprod{x_1}{x_2}|$ is at most the number of collisions between~$x_1$ and~$x_2$.

We will construct $h$ in such a way that for all pairs $x_1, x_2$, the number of collisions is small, say at most $C$. Then by setting $t > 2C$, we ensure that $|t \mult \OPT - \OPT'| < t / 2$ so we can recover $t \mult \OPT$ by computing $\OPT'$ and rounding to the closest multiple of $t$. In particular, the optimal pair of vectors in the hashed instance correspond to the pair with maximum inner product in the original instance. Note that we crucially use the fact that \MaxIP{}
is expressive enough to \emph{compute the value} of the inner product, which allows us to get rid of the small additive error introduced by the hashing (after rounding).

In~\autoref{sec:universe-reduction} we show that the desired hash function exists and is in fact
deterministic: Pick any $t$ primes $p_1,\dots,p_t$ of size $\Theta(t\log t)$ and let $d' = p_1 + \dots + p_t$. We identify $[d']$ with $\{ (i, p_j) : 1 \leq j \leq t, 0 \leq i < p_j \}$ and assign the auxiliary vector $w_i$ to have one-entries exactly at all coordinates $(i \bmod p_j, p_j)$, $1 \leq j \leq t$. A simple calculation shows that with this construction the number of collisions between $x_1$ and $x_2$ is at most $\norm{x_1}_1 \mult \norm{x_2}_1 \mult \log d$, see~\autoref{lem:bloom-filter-deterministic}. With some additional tricks, we can control this quantity.

Our analysis allows us to even maintain $c$-approximate solutions, albeit with an arbitrarily small blow-up due to the small error introduced by rounding. Finding a fully approximation-preserving reduction remains a challenge for future work. Additionally, we need to take great care that our reductions are efficient enough to even transfer $\log^{0.1} n$-improvements, to obtain our speed-up for exact optimization (\autoref{thm:lower-order-exact}).

\paragraph*{Comparison to Gao et al.'s Work}
Our reduction is similar to the work of Gao, Impagliazzo, Kolokolova and Williams~\cite{GaoIKW18}, showing that the sparse version of Orthogonal Vectors is complete for model checking first-order properties. Here we discuss the key differences.

This first phase of our reduction follows the same structure as in Gao et al., but we simplify the proof significantly: One major difference is that they define a more complicated version of the Hybrid Problem including cross predicates \cite[Section 5.2]{GaoIKW18}. Borrowing ideas from~\cite{BringmannFK19}, we remove the cross predicates at an earlier stage of the reduction (Step 2), which simplifies the remaining Steps~3 and~4. The absence of cross predicates also simplifies the baseline algorithm (\autoref{sec:baseline}). More generally, by splitting the reduction into a chain of four steps we cleanly separate the main technical ideas used in the first phase; see \autoref{sec:reduction-to-hybrid} for more details. In the same spirit we simplify Gao et al.'s improved algorithm \cite[Section 9.2]{GaoIKW18} for all problems with more than 1 counting quantifier avoiding their case distinction of 9 different cases by using a simple basis to represent all Boolean functions $\phi: \{0,1\}^3 \to \{0,1\}$; see~\autoref{sec:improved-alg-multiple-counting}.

In the second phase of the reduction, their work faces the same main challenge as ours. Specifically, reducing their Hybrid Problem to OV naively requires complementing a sparse binary relation, possibly resulting in a large dense complement. They solve this issue by designing a similar dimension reduction as ours using a Bloom filter. Naturally their dimension reduction is randomized, but they also provide a derandomization. However, note that there is a crucial difference: They reduce to OV which is a \emph{decision} problem, while we reduce to the \emph{optimization} problem \MaxIP{}. For this reason, the dimension reductions differ in nature: One the one hand, we exploit that \MaxIP{} is more expressive than OV -- namely that \MaxIP{} can handle a small number of errors if we round the result, while for OV any introduced error would result in vectors that are not orthogonal anymore. On the other hand, by reducing to OV, Gao et al.\ do not have to worry about ``compressing'' the range of possible optimal values, or making the reduction approximation-preserving. For these reasons, their dimension reduction would be unsuitable in our work, and ours would be unsuitable in their work.
\section{The Reduction}
In this section we give the proofs of our main results. The following lemma captures our reduction in all generality. Let \kMaxIP{k} denote the generalization of the \MaxIP{} problem with the objective to compute $\max_{x_1 \in X_1, \dots, x_k \in X_k} \langle x_1, \dots, x_k \rangle$, where $\langle x_1, \dots, x_k \rangle = \sum_y x_1[y] \mult \dots \mult x_k[y]$. We define \kMinIP{k} analogously.

\begin{lemma} \label{lem:reduction}
Let $s(n) \leq n^{1/6}$ be a nondecreasing function and let $c \geq 1$ be constant. Assume that $\kMaxIP{k}$ in dimension $d = \widetilde\Order(s(n)^4 \log^2 n)$ can be $c$-approximated in time $\Order(n^k / s(n))$, and let $\psi$ be an arbitrary $\MaxSP_k$ formula.
\begin{itemize}
\item If $c = 1$ (i.e., we are in the case of exact computation), then $\Max(\psi)$ can be exactly solved in time $\Order(m^k / s(\!\sqrt[k+1]m))$.
\item If $c > 1$, then $\Max(\psi)$ can be $(c + \varepsilon)$-approximated in time $\Order(m^{k} /s(\!\sqrt[k+1]m))$, for any constant $\varepsilon > 0$. 
\end{itemize}
The analogous statement holds for $\kMinIP{k}$ and $\MinSP_k$.
\end{lemma}

The outline for this section is as follows. First we show how to derive the completeness result (Theorems~\ref{thm:completeness-ip} and~\ref{thm:completeness-ip-approx}) and the lower-order improvements (Theorems~\ref{thm:lower-order-exact} and~\ref{thm:lower-order-approx}) from \autoref{lem:reduction} in~\autoref{sec:consequences}. Then we present the proof of \autoref{lem:reduction}, which is carried out in two phases as explained in the technical overview. In \autoref{sec:hybrid-problem} we formally introduce the intermediate problem called the \emph{Hybrid Problem}. In \autoref{sec:universe-reduction} we give a fine-grained reduction from the Hybrid Problem to Maximum or Minimum Inner Product (\autoref{lem:hybrid-to-ip}). Finally, in \autoref{sec:reduction-to-hybrid} we reduce any $\Opt_{k,\ell}$ formula to the Hybrid Problem (\autoref{lem:optsp-to-hybrid}), thus finishing the proof of \autoref{lem:reduction}. We will pay particularly close attention to the exact savings $s$ in every step.

\subsection{Consequences}\label{sec:consequences}
First we derive the completeness Theorems~\ref{thm:completeness-ip} and~\ref{thm:completeness-ip-approx} from \autoref{lem:reduction}.

\begin{proof}[Proof of Theorems~\ref{thm:completeness-ip} and~\ref{thm:completeness-ip-approx}]
Let $c \geq 1$ denote the approximation ratio (that is, $c = 1$ for \autoref{thm:completeness-ip} and $c \geq 1$ for \autoref{thm:completeness-ip-approx}). Assuming that \SparseMaxIP{} can be $c$-approximated in time $O(m^{2-\delta})$ for some $\delta > 0$, we obtain an algorithm for $c$-approximating \MaxIP{} in dimension $d = n^{4\delta/9}$ in time $\Order((nd)^{2-\delta}) = \Order(n^{2-\delta}d^2) = \Order(n^{2-\delta/9})$. We also obtain an algorithm for $c$-approximating \kMaxIP{k} in the same dimension in time $\Order(n^{k-\delta/9})$ (brute-force all options for the first $k-2$ vectors, then use the \kMaxIP{2} algorithm). We can now plug this improved algorithm into our reduction: Setting $s(n) = n^{\delta / 9} / \polylog(n)$ we have that \kMaxIP{k} in dimension $d = \widetilde\Order(s(n)^4 \log^2 n)$ can be $c$-approximated in time $\Order(n^k / s(n))$. Thus, if $c = 1$ we obtain by \autoref{lem:reduction} that $\Opt(\psi)$ can be exactly solved in time $\Order(m^k / s(\!\sqrt[k+1]m)) = \Order(m^{k-\beta})$ for $\beta = \frac{\delta}{9(k+1)} > 0$. If $c > 1$, we obtain that $\Opt(\psi)$ can be $(c+\varepsilon)$-approximated in the same running time, for an arbitrarily small constant~$\varepsilon > 0$.
\end{proof}

Next, we prove~\autoref{thm:equivalence-ip}.

\begin{proof}[Proof of \autoref{thm:equivalence-ip}]
The reductions from \SparseMaxIP{} to \MaxIP{} and from \SparseMinIP{} to \MinIP{} for both the exact and approximate settings are a direct consequence of~\autoref{lem:reduction}.

For the other direction, assume there exists some $\delta > 0$ such that \SparseMaxIP{} can be $c$-approximated in time $\Order(m^{2-\delta})$. Set $\gamma := \delta/2$ and observe that any \MaxIP{} instance with $d=n^\gamma$ yields a \SparseMaxIP{} instance of size $m=O(nd) = O(n^{1+\gamma})$. Since we can solve this instance in time $O(m^{2-\delta}) = O(n^{(1+\gamma)(2-\delta)}) = O(n^{(1+\delta/2)(2-\delta)}) = O(n^{2-\delta^2/2})$, we obtain a $O(n^{2-\delta'})$-algorithm for \MaxIP{} with $d = n^\gamma$ and $\delta' = \delta^2/2$. Note that this works for both the exact ($c = 1$)
    and approximate ($c > 1$) settings. The proof for the minimization case is analogous.
\end{proof}

To prove Theorems~\ref{thm:lower-order-exact} and~\ref{thm:lower-order-approx}, we make use of the following state-of-the-art algorithms for \MaxIP{} and \MinIP{}, established in three previous papers~\cite{AlmanCW16,Chen18,ChenW19}.

\begin{theorem}[Improved Algorithms for \MaxIP{} and \MinIP{}~\cite{AlmanCW16,Chen18,ChenW19}]\label{thm:improved-algorithms-ip}
\hspace{0cm}
\begin{itemize}
\item \kMaxIP{k} and \kMinIP{k} in dimension $d = \Order(\log^{2.9} n)$ can be exactly solved in randomized time $\Order(n^k / \log^{100} n)$~\emph{\cite{AlmanCW16}}.
\item For any constant $c > 1$, \kMaxIP{k} in dimension $d = 2^{\Order(\sqrt{\log n})}$ can be $c$-approximated in deterministic time \raisebox{0pt}[0pt][0pt]{$n^k / 2^{\Omega(\sqrt{\log n})}$}~\emph{\cite[Theorem~1.5]{Chen18}}.
\item For any constant $c > 1$, \kMinIP{k} in dimension $d = 2^{\Order(\sqrt{\log n})}$ can be $c$-approximated in randomized time \raisebox{0pt}[0pt][0pt]{$n^k / 2^{\Omega(\sqrt{\log n})}$}~\emph{\cite[Theorem~1.7]{ChenW19}}.
\end{itemize}
\end{theorem}

\begin{proof}[Proof of Theorems~\ref{thm:lower-order-exact} and~\ref{thm:lower-order-approx}]
To prove \autoref{thm:lower-order-exact}, we plug in the first algorithm from \autoref{thm:improved-algorithms-ip} into \autoref{lem:reduction} and choose $s(n) = \log^{0.1}n$. We obtain an exact $\OptSP_k$ algorithm in time $m^k / \log^{\Omega(1)} m$.

For \autoref{thm:lower-order-approx}, we plug the second and third algorithms from \autoref{thm:improved-algorithms-ip} into \autoref{lem:reduction} and choose \raisebox{0pt}[0pt][0pt]{$s(n) = 2^{\Order(\sqrt{\log n})}$}. We get a $c$-approximation for $\OptSP_k$ in time \raisebox{0pt}[0pt][0pt]{$m^k / 2^{\Omega(\sqrt{\log m})}$}, for any constant $c > 1$.
\end{proof}

Note that only one of these algorithms is deterministic; other known deterministic algorithms are not efficient enough for our reduction\footnotemark.
\footnotetext{Focus on exact $\MaxSP_k$ for illustration: To obtain the same savings as in \autoref{thm:lower-order-exact}, we would need a deterministic algorithm for \MaxIP{} in dimension $d = \Order(\log^{2.9} n)$ running in time $\Order(n^2 / \log^{100} n)$. However, for this speed-up the current best algorithm~\cite{AlmanCW16} requires $d = \Order(\log^{1.9} n)$, so one needs to either improve the algorithm or improve our dimension reduction (\autoref{lem:reduction}) to dimension $d = \poly(s(n)) \log n$, say.}

\subsection{The Hybrid Problem} \label{sec:hybrid-problem}
We start with another problem definition.

\begin{definition}[Basic Problem]
Given set families $\mathcal S_1, \dots, \mathcal S_k$ over a universe $U$, the \emph{Basic Maximization Problem of type $\tau \in \{0, 1\}^k$} is to to compute
\begin{equation*}
    \OPT = \max_{S_1 \in \mathcal S_1, \dots, S_k \in \mathcal S_k} \Bigg| \Bigg(\bigcap_{i : \tau[i] = 1} S_i\Bigg) \setminus \Bigg(\bigcup_{i : \tau[i] = 0} S_i\Bigg) \Bigg|.
\end{equation*}
\end{definition}

For example, the Basic Problem of type $\tau = 11$ is to maximize the common intersection of two sets $S_1$ and $S_2$, the Basic Problem of type $\tau = 10$ is to maximize the number of elements in $S_1$ not contained in $S_2$ and the Basic Problem of type $\tau = 00$ is to maximize the number of universe elements contained in neither $S_1$ nor $S_2$.

Note that every Basic Problem can be seen as an $\OptSP_k$ formula: We introduce objects for all sets $S_i$ and all universe elements $u$, and connect $S_i$ to $u$ via an edge $E(S_i, u)$ if and only if $u \in S_i$. Consistent with this analogy, we define $n$ as the total number of sets $S_i$ and $m$ as the total cardinality of all sets $S_i$ and, as before, study the Basic Problem with respect to the sparsity $m$.

\begin{definition}[Hybrid Problem]
Given set families $\mathcal S_1, \dots, \mathcal S_k$ over a universe $U$, which is partitioned into $2^k$ parts $U = \bigcup_{\tau \in \{0, 1\}^k} U_\tau$, the \emph{Hybrid Maximization Problem} is to compute
\begin{equation*}
    \OPT = \max_{S_1 \in \mathcal S_1, \dots, S_k \in \mathcal S_k} \sum_{\tau \in \{0, 1\}^k} \Bigg| \,U_\tau \cap \Bigg(\bigcap_{i : \tau[i] = 1} S_i\Bigg) \setminus \Bigg(\bigcup_{i : \tau[i] = 0} S_i\Bigg) \Bigg|.
\end{equation*}
\end{definition}

We similarly define Basic Minimization Problems and define $c$-approximations of Basic Problems in the obvious way. For any $S_1,\dots,S_k$ and $\tau \in \{0, 1\}^k$ we denote by $\Val_{\tau}(S_1,\dots,S_k)$ the \emph{value} of the Basic Problem constraint of type $\tau$:
\begin{equation*}
    \Val_{\tau}(S_1,\dots,S_k) := \Bigg| \,U_\tau \cap \Bigg(\bigcap_{i : \tau[i] = 1} S_i\Bigg) \setminus \Bigg(\bigcup_{i : \tau[i] = 0} S_i\Bigg) \Bigg|.
\end{equation*}
And we use $\Val(S_1, \dots, S_k) := \sum_{\tau} \Val_{\tau}(S_1, \dots, S_k)$ to denote the total value of the sets $S_1, \dots, S_k$ in a Hybrid Problem instance.

Intuitively, the Hybrid Problem simultaneously optimizes Basic Problem constraints of different types. If we could afford to complement (parts of) the sets $S_i$, then there is a straightforward reduction from the Hybrid Problem to a Basic Problem of arbitrary type $\tau$: For each constraint of type $\tau' \neq \tau$, we simply complement all sets $S_i$ with $\tau[i] \neq \tau'[i]$ (more precisely, construct sets $S_i'$ such that $U_{\tau'} \cap S_i' = U_{\tau'} \setminus S_i$) and reinterpret the $\tau'$-constraint as type~$\tau$. In summary:

\begin{observation} \label{obs:hybrid-to-basic}
In time $\Order(n |U|)$, any Hybrid Problem instance can be converted into an equivalent Basic Problem instance of arbitrary type $\tau$. The sparsity of the constructed instance is up to $n |U|$.
\end{observation}

However, being in the sparse setup we cannot tolerate the blow-up in the sparsity. Therefore, in order to efficiently apply \autoref{obs:hybrid-to-basic}, we first have to control the universe size $|U|$. 

\subsection{Universe Reduction}\label{sec:universe-reduction}
The goal of this section is to reduce the Hybrid Problem to \kMaxIP{k}. We give a reduction which closely preserves the savings $s(n)$ achieved by exact or approximate \kMaxIP{k} algorithms (losing only polynomial factors in $s(n)$). As a drawback, the reduction slightly worsens the approximation factor, turning a $c$-approximation into a $(c+\varepsilon)$-approximation.

\begin{lemma} \label{lem:hybrid-to-ip}
Let $s(n) \leq n^{1/6}$ be a nondecreasing function and assume that $\kMaxIP{k}$ in dimension $d = \widetilde\Order(s(n)^4\log^2 n)$ can be $c$-approximated in time $\Order(n^k / s(n))$.
\begin{itemize}
\item If $c = 1$ (i.e., we are in the case of exact computation), then the Hybrid Problem can be exactly solved in time $\Order(m^k / s(m))$.
\item If $c > 1$, then the Hybrid Problem can be $(c + \varepsilon)$-approximated in time $\Order(m^{k} / s(m))$, for any constant $\varepsilon > 0$. 
\end{itemize}
The analogous statement holds for \kMinIP{k} and $\MinSP_k$.
\end{lemma}

On a high level, we prove \autoref{lem:hybrid-to-ip} by first using a deterministic construction to reduce the universe size, and then reducing further to \kMaxIP{k} as in \autoref{obs:hybrid-to-basic}. The following lemma provides our universe reduction in the form of a hash-like function $h$.

\begin{lemma} \label{lem:bloom-filter-deterministic}
Let $U$ be a universe and let $t$ be a parameter. There exists a universe $U'$ of size at most $4 t^2 \log t$ and a function $h$ mapping elements in $U$ to size-$t$ subsets of $U'$, such that the following properties hold. By abuse of notation, we write $h(S) = \bigcup_{u \in S} h(u)$ for sets $S \subseteq U$.

\begin{enumerate}
\item (Hashing.) For all sets $S \subseteq U$, it holds that $|h(S)| \geq t |S| - |S|^2 \log |U|$.
\item (Efficiency.) Evaluating $h(u)$ takes time $\widetilde\Order(t)$.
\end{enumerate}
\end{lemma}

\begin{proof}
We start with the construction of $h$. By the Prime Number Theorem, there exist~$t$ primes $p_1, \dots, p_t$ in the interval $[2t \log t, 4t \log t]$ (for large enough $t$, see~\cite[Corollary~3]{RosserS62} for the quantitative version). Let $U' = \{ (i, j) : 1 \leq i \leq t, 0 \leq j < p_i \}$, then $|U'| \leq 4 t^2 \log t$. We identify $U$ with $[|U|]$ in an arbitrary way and define $h(u) = \{ (i, u \bmod p_i) : 1 \leq i \leq t \}$ for~$u \in U$.

In order to prove the first property, let us define the \emph{collision number} of two distinct elements $u, u' \in U$ as $|h(u) \cap h(u')|$. It is easy to see that the collision number of any such pair is at most $\log |U|$: For any prime $p_i$, we have that $u \bmod p_i = u' \bmod p_i$ if and only if $p_i$ divides $u - u'$. Since $u - u'$ has absolute value at most $U$, there can be at most $\log |U|$ distinct prime factors $p_i$ of $u - u'$. It follows that $t |S| - |h(S)| \leq \sum_{u, u' \in S} |h(u) \cap h(u')| \leq |S|^2 \log |U|$.

Finally, the function can be efficiently evaluated: Computing the primes $p_1, \dots, p_t$ takes time $\Order(t \log t \log \log t)$ using Eratosthenes' sieve, for example. After this precomputation, evaluating $h(u)$ in time $\Order(t)$ is straightforward.
\end{proof}

\begin{lemma}[Universe Reduction] \label{lem:dimensionality-reduction-deterministic}
Let $\mathcal S_1, \dots, \mathcal S_k$ over the universe $U = \bigcup_\tau U_\tau$ be a Hybrid Problem instance of maximum set size $s = \max_{S_i \in \mathcal S_i} |S_i|$, and let $t$ be a parameter. In time $\widetilde\Order(mt)$ we can compute a number $\Delta \geq 0$ and a new Hybrid Problem instance $\mathcal S_1', \dots, \mathcal S_k'$ over a small universe $U' = \bigcup_\tau U_\tau'$ of size $|U'| = \Order(t^2 \log t)$ such that:
\begin{enumerate}
\item\label{lem:dimensionality-reduction-deterministic:itm:1} The sets $S_i \in \mathcal S_i$ and the sets $S_i' \in \mathcal S_i'$ stand in one-to-one correspondence.
\item\label{lem:dimensionality-reduction-deterministic:itm:2} For all $S_1 \in \mathcal S_1, \dots, S_k \in \mathcal S_k$, it holds that:
\begin{equation*}
    \left| t \mult \Val(S_1, \dots, S_k) - \Val(S_1', \dots, S_k') - \Delta \right| = \Order(s^2 \log |U|).
\end{equation*}
\end{enumerate}
\end{lemma}

\begin{proof}
We first describe how to construct the new instance. The first goal is to design individual universe reductions for all subuniverses $U_\tau$, that is, we construct new universes~$U_\tau'$ and functions $h_\tau$ mapping $U_\tau$ to size-$t$ subsets of $U_\tau$. We distinguish two cases:
\begin{itemize}
\item If $|U_\tau| \leq 4 t \log t$, then we simply take $U_\tau'$ as $t$ copies of $U_\tau$ and let $h_\tau$ be the function which maps any element to its $t$ copies in $U_\tau$. It holds that $|U_\tau'| = t \mult |U_\tau| \leq 4t^2 \log t$.
\item If $|U_\tau| > 4 t \log t$, then we apply \autoref{lem:bloom-filter-deterministic} with parameter $t$ to obtain $U_\tau'$ and $h_\tau$. The lemma guarantees that $|U_\tau'| \leq 4 t^2 \log t$.
\end{itemize}
Next, we assemble these individual reductions into one. Set $U' = \bigcup_\tau U_\tau'$, where we treat the sets $U_\tau'$ as disjoint. Since in both of the previous two cases we have $|U_\tau| = \Order(t^2 \log t)$ it follows that $|U| = \sum_\tau |U_\tau| = \Order(t^2 \log t)$. Let $h$ be the function which is piece-wise defined by the $h_\tau$'s, that is, $h$ returns $h_\tau(u)$ on input~$u \in U_\tau$. Recall the notation $h(S) = \bigcup_{u \in S} h(u)$. The new Hybrid Problem instance is constructed by hashing every set $S_i \in \mathcal S_i$ into the smaller universe, that is, we set $S_i' := h(S_i) \in \mathcal S_i'$. Property~\ref{lem:dimensionality-reduction-deterministic:itm:1} is immediate from this construction, and the computation takes time $\widetilde\Order(mt)$.

It remains to prove Property~\ref{lem:dimensionality-reduction-deterministic:itm:2}. For the remainder of the proof fix some sets $S_1, \dots, S_k$ and let $S = S_1 \cup \dots \cup S_k$ (clearly,~$S$ has size $\Order(s)$). We start with the (unrealistic) assumption that $S$ is hashed \emph{perfectly}, that is, $|h(S)| = t |S|$. In this case we claim that:
\begin{itemize}
\item $t \mult \Val_\tau(S_1, \dots, S_k) = \Val_\tau(h(S_1), \dots, h(S_k))$ for all $\tau \neq 0^k$,
\item $t \mult \Val_\tau(S_1, \dots, S_k) = \Val_\tau(h(S_1), \dots, h(S_k)) + \Delta$ for $\tau = 0^k$, where $\Delta := t \mult |U_{0^k}| - |U_{0^k}'|$. 
\end{itemize}
Indeed, if $S$ is hashed perfectly then we exactly scale the number of satisfying elements by a factor of $t$ for every type $\tau \neq 0^k$. This holds because a satisfying assignment for $\tau \neq 0^k$ corresponds to some element of the universe $u \in U_{\tau}$ for which $u \in S_i$ for all $i$'s such that $\tau[i] = 1$. The perfect hashing implies that the element $u$ in these sets $S_i$ gets mapped to $t$ different elements in the new universe, and since there are no collisions these form $t$ satisfying assignments in the hashed instance. The type $\tau = 0^k$ is exceptional because each satisfying assignment does not correspond to any $u \in U_{0^k}$. Instead, the hashing scales the number of \emph{falsifying} elements of type $0^k$, $|U_{0^k} \cap S|$. The number of \emph{satisfying} elements of type $0^k$, $|U_{0^k} \setminus S|$, is preserved up to an additive error of exactly $\Delta$.

We will now remove the unrealistic assumption that $h$ is hashed perfectly. The strategy is to define another function $h^*$ obtained from $h$ by artificially making the hashing with $S$ perfect. To that end, we list the elements in $S$ in an arbitrary order $s_1, \dots, s_{|S|}$, and start with the assignment $h^*(s_j) = h(s_j)$. As long as there exist indices $i < j$ such that $h^*(s_i)$ and $h^*(s_j)$ share a common element $z$, we reassign $h^*(s_j) := h^*(s_j) \setminus \{z\} \cup \{z'\}$ for some unused universe element $z' \in U'$. The function $h^*$ obtained in this way also maps elements of $U$ to size-$t$ subsets of $U'$ and hashes $S$ perfectly. Let $Z$ be the set of all pairs of elements $z$ and $z'$ that occurred in the process; since there are exactly $t|S| - |h(S)|$ iterations we have $|Z| \leq 2t|S| - 2|h(S)|$ and by \autoref{lem:bloom-filter-deterministic} it follows that $|Z| = \Order(s^2 \log |U|)$. By the definition of $h^*$, it is clear that $|\Val(h(S_1), \dots, h(S_k)) - \Val(h^*(S_1), \dots, h^*(S_k))| \leq |Z|$. Therefore, by the previous paragraph (applied with $h^*$) and by an application of the triangle inequality, we obtain:
\begin{itemize}
\item $|t \mult \Val_\tau(S_1, \dots, S_k) - \Val_\tau(h(S_1), \dots, h(S_k))| = \Order(s^2 \log |U|)$ for all $\tau \neq 0^k$,
\item $|t \mult \Val_\tau(S_1, \dots, S_k) - \Val_\tau(h(S_1), \dots, h(S_k)) - \Delta| = \Order(s^2 \log |U|)$ for $\tau = 0^k$.
\end{itemize}
The claimed Property~\ref{lem:dimensionality-reduction-deterministic:itm:2} is now immediate by summing over all types~$\tau$ and by another application of the triangle inequality.

Finally, it remains to prove that $\Delta \geq 0$. There are two cases depending on how the set~$U_{0^k}'$ was constructed: In the first case of the construction we have $t \mult |U_{0^k}| = |U_{0^k}'|$ and thus $\Delta = 0$. In the second case we have $|U_\tau'| \leq 4 t^2 \log t < t \mult |U_\tau|$ and thus $\Delta = t \mult |U_{0^k}| - |U_{0^k}'| > 0$.
\end{proof}

Having established the universe reduction, we can finally prove \autoref{lem:hybrid-to-ip}.
\newcommand\ALG{\operatorname{ALG}}

\begin{proof}[Proof of \autoref{lem:hybrid-to-ip}]
The algorithm consists of three steps, which are implemented in the same way for all combinations of maximization versus minimization and exact versus approximate computation.

\begin{enumerate}
\itemsep0.5em
\item \emph{(Eliminating heavy sets.)} We say that a set $S_i \in \calS_i$ is \emph{heavy} if $|S_i| > s(m)$, and \emph{light} otherwise. Our first goal is to eliminate all heavy sets. Since the total cardinality of all sets $S_i$ is bounded by $m$, there can be at most $O(m/s(m))$ heavy sets. Therefore, we can brute-force over every such set $S_i$ and solve the remaining Hybrid Problem on $k-1$ set families using the baseline algorithm in time $O(m^{k-1})$. Afterwards, we can safely remove all heavy sets. Overall, this step takes time $O(m^k/s(m))$.

\item \emph{(Reduction to \kMaxIP{k} or \kMinIP{k}.)} In the remaining instance we have that $|S_i| \leq s(m)$ for all sets~$S_i$. Therefore, we can apply the universe reduction from \autoref{lem:dimensionality-reduction-deterministic} (with some parameter $t$ to be specified in the next step) to obtain an instance $\calS'_1,\dots,\calS'_k$ over a smaller universe $U' = \bigcup_\tau U'_{\tau}$ of size $\Order(t^2 \log t)$, and an offset $\Delta \geq 0$.

The Hybrid Maximization Problem instance $\calS'_1,\dots,\calS'_k$ reduces to $\kMaxIP{k}$ in the natural way: Recall that \kMaxIP{k} is the same as the Basic Problem of type $\tau = 1^k$. Hence, we can apply \autoref{obs:hybrid-to-basic} to reduce to an instance of \kMaxIP{k} with $n = \Order(m)$ vectors in dimension $\Order(t^2 \log t)$ in time $\Order(n |U'|) = \Order(n t^2 \log t)$. An analogous reduction works for Hybrid Minimization Problems and \kMinIP{k}.

\item \emph{(Recovering the optimal value.)} Solve (or approximate) the constructed \kMaxIP{k} instance and let $\ALG'$ denote the output. Then compute $\ALG := (\ALG' + \Delta) / t$ and return $\ALG$ rounded to an integer. The precise way of rounding depends on maximization versus minimization and exact versus approximate, see the following analysis.
\end{enumerate}

\noindent Let $\varepsilon > 0$ be a constant which we will specify later, and set $t = C s(m)^2 \log m$ for some sufficiently large constant $C = C(\varepsilon)$. Then by Property~\ref{lem:dimensionality-reduction-deterministic:itm:2} of \autoref{lem:dimensionality-reduction-deterministic} we have
\begin{equation*}
    \left| \Val(S_1, \dots, S_k) - \frac{\Val(S_1', \dots, S_k') + \Delta}t \right| = \Order\left(\frac{s(m)^2 \log m}t\right) < \varepsilon.
\end{equation*}
In particular, it holds that 
\begin{equation} \label{eqn:ineq-opt}
    \left| \OPT - \frac{\OPT' + \Delta}t \right| < \varepsilon,
\end{equation}
where $\OPT$ and $\OPT'$ are the optimal values of the original and the reduced instance, respectively. As the new universe has size $\Order(t^2 \log t) = \widetilde\Order(s(m)^4 \log m^2)$ as claimed, we can indeed use the efficient $\Order(m^k / s(m))$-time \kMaxIP{k} or \kMinIP{k} algorithm in the third step. The total running time is as stated: Recall that $s(m) \leq m^{1/6}$ and thus all previous steps run in time $\Order(m^k / s(m))$. It remains to argue about the guarantees of the reduction; we need to consider three cases:
\begin{itemize}
\item\emph{(Exact maximization or minimization: $c = 1$.)} It suffices to set $\varepsilon < \frac12$. Since we can exactly compute $\ALG' = \OPT'$, by rounding $\ALG = (\ALG'+\Delta)/t$ to the nearest integer, we obtain the only integer in the interval $((\OPT' + \Delta) / t - \frac12, (\OPT' + \Delta) / t + \frac12)$, and thus we output $\OPT$.
\item\emph{(Approximate maximization: $c > 1$.)} We have $c^{-1} \OPT' \leq \ALG' \leq \OPT'$ and therefore
\begin{align*}
    \ALG &= \frac{\ALG' + \Delta}t \leq \frac{\OPT' + \Delta}t \stackrel{{\eqref{eqn:ineq-opt}}}{\leq} \OPT + \varepsilon, \\
    \ALG &= \frac{\ALG' + \Delta}t
        \geq \frac{c^{-1}(\OPT' + \Delta)}t 
        \stackrel{{\eqref{eqn:ineq-opt}}}{\geq} c^{-1} (\OPT - \varepsilon)
        \geq c^{-1} \OPT - \varepsilon,
\end{align*}
where in the first inequality of the second line we used both $\ALG' \geq c^{-1}\OPT'$ and $c > 1$. From these bounds we derive that the algorithm should return $\ceil{\ALG - \varepsilon}$. Indeed, as $\ceil{\ALG - \varepsilon} \leq \OPT$ this is always a feasible solution. Moreover, the solution is $\frac c{1-2\varepsilon}$-approximate: If $\OPT = 0$, then $\ceil{\ALG - \varepsilon} = 0$ (if we set $\varepsilon < \frac12$). If $\OPT \geq 1$, then $\ceil{\ALG - \varepsilon} \geq \frac1c \OPT - 2\varepsilon \geq \frac{1-2\varepsilon}c \OPT$. Setting~$\varepsilon$ small enough yields approximation ratio $c + \varepsilon'$, for any $\varepsilon' > 0$.
\item \emph{(Approximate minimization: $c > 1$.)} We have $\OPT' \leq \ALG' \leq c \mult \OPT'$ and therefore
\begin{align*}
    \ALG &= \frac{\ALG' + \Delta}t \leq \frac{c \mult \OPT' + \Delta}t 
        \leq \frac{c \mult (\OPT + \Delta)}t 
        \stackrel{{\eqref{eqn:ineq-opt}}}{\leq} c \mult \OPT + c \mult \varepsilon, \\
    \ALG &= \frac{\ALG' + \Delta}t \geq \frac{\OPT' + \Delta}t 
        \stackrel{{\eqref{eqn:ineq-opt}}}{\geq} \OPT - \varepsilon.
\end{align*}
In this case the algorithm should return $\floor{\ALG + \varepsilon}$. This solution is always feasible as $\OPT \leq \floor{\ALG + \varepsilon}$. Moreover, the solution is $c(1+2\varepsilon)$-approximate: If $\OPT = 0$, then $\floor{\ALG + \varepsilon} = 0$ (if we set $\varepsilon < \frac12$). If $\OPT \geq 1$, then $\floor{\ALG + \varepsilon} \leq c \mult \OPT + (c+1)\varepsilon \leq c(1+2\varepsilon) \OPT$. We may again set $\varepsilon$ small enough to obtain approximation ratio $c + \varepsilon'$, for any $\varepsilon' > 0$. \qedhere
\end{itemize}
\end{proof}

\subsection{Reducing \texorpdfstring{\boldmath$\OptSP_k$}{OptSPk} Formulas to the Hybrid Problem}\label{sec:reduction-to-hybrid}
In this section we give the first phase of the reduction, where we reduce $\Opt(\psi)$ to the Hybrid Problem. The main lemma is the following. As before, let $s(m) \leq m^{1/6}$ be a nondecreasing function and let $c \geq 1$ be constant.

\begin{lemma} \label{lem:optsp-to-hybrid}
Let $k \geq 2$. If the Hybrid Problem can be $c$-approximated in time $\Order(m^k / s(m))$, then $\Opt(\psi)$ can be $c$-approximated in time $\Order(m^k / s(\!\sqrt[k+1]m))$, for any $\OptSP_k$ formula~$\psi$.
\end{lemma}

Recall that we only have to deal with $\OptSP_k = \OptSP_{k,1}$ formulas, as any $\OptSP_{k,\ell}$ problem with $\ell > 1$ directly admits an improved algorithm; see \autoref{sec:improved-alg-multiple-counting}. As explained in~\autoref{sec:technical-overview}, we prove \autoref{lem:optsp-to-hybrid} by progressively simplifying $\Opt(\psi)$ in four steps:
\begin{enumerate}
\item Remove all \emph{hyperedges}, that is, $\psi$ no longer contains predicates of arity $\geq 3$ so an instance of $\Opt(\psi)$ can be thought of as a (colored) graph.
\item Remove all \emph{cross edges}, that is, edges between vertices $x_i$ and $x_j$ that we maximize over.
\item Remove all \emph{parallel edges} (or alternatively, colored edges), that is, we combine parallel edges into simple edges.
\item Remove unary predicates, finally turning the $\Opt(\psi)$ instances into graphs. At this point it becomes simple to rewrite $\Opt(\psi)$ as a Hybrid Problem.
\end{enumerate}

\paragraph*{Step 1: Removing Hyperedges}
As a first step, we eliminate all \emph{hyperpredicates}, that is, predicates of arity $\geq 3$. Formally, we prove the following lemma.

\begin{lemma} \label{lem:hyperedges}
Suppose that, for any $\OptSP_k$ formula $\psi$ not containing hyperpredicates, $\Opt(\psi)$ can be $c$-approximated in time $\Order(m^k / s(m))$. Then $\Opt(\psi)$ can be $c$-approximated in time $\Order(m^k / s(m))$ for any $\OptSP_k$ formula $\psi$.
\end{lemma}

The proof is quite similar to~\cite[Section 7]{GaoIKW18}. We start with a technical lemma:

\begin{lemma} \label{lem:positive-cross-edge}
Let
\begin{equation*}
    \psi = \opt_{x_1, \dots, x_k} \counting_y \Big( E(x_i, x_j) \land \phi(x_1, \dots, x_k, y) \Big),
\end{equation*}
for some $i, j \in [k], i \neq j$ and arbitrary $\phi$. Then $\Opt(\psi)$ can be solved exactly in time $\Order(m^{k-1/2})$.
\end{lemma}
\begin{proof}
Let us begin with the simplest case $k = 2$. For a vertex $x$ in the given instance, let $\deg(x)$ denote the total number of records containing~$x$ over all relations. We distinguish the following three cases:
\begin{description}
\item[Case 1: \boldmath$\deg(x_1) \geq \sqrt m$.] We explicitly list all vertices $x_1$ with $\deg(x_1) \geq \sqrt m$; note that there can be at most $\Order(\sqrt m)$ such elements since the sparsity of the $\Max(\psi)$ instance is bounded by $\Order(m)$. The remaining $\MaxSP_1$ formula can be solved in time $\Order(m)$ using the baseline algorithm. In total, this step takes time $\Order(\sqrt m m) = \Order(m^{3/2})$.
\item[Case 2: \boldmath$\deg(x_2) \geq \sqrt m$.] By exchanging the roles of $x_1$ and $x_2$, we deal with this case in the same way as case 1.
\item[Case 3: \boldmath$\deg(x_1) < \sqrt m$ and $\deg(x_2) < \sqrt m$.] Assuming that the previous two cases were executed, we can assume that $\deg(x_1) < \sqrt m$ and $\deg(x_2) < \sqrt m$ for all remaining objects~$x_1, x_2$. We exploit that any non-zero solution $(x_1, x_2)$ of $\Max(\psi)$ satisfies $E(x_1, x_2)$: It suffices to maximize over all $\Order(m)$ edges $E(x_1, x_2)$, counting the number of~$y$'s satisfying~$\phi$. Since $\deg(x_1) < \sqrt m$ and $\deg(x_2) < \sqrt m$, we can enumerate and test all objects $y$ which are connected to either $x_1$ or $x_2$ by some relation in time $\Order(\sqrt m)$. What remains are objects $y$ not connected to either $x_1$ or $x_2$ by any relation. To account for these missing objects, we can substitute $\false$ for all non-unary predicates in $\phi$; what remains is a Boolean function over unary predicates over~$y$. We can precompute the number of $y$'s satisfying that function in linear time, so again the total time is $\Order(m + m \sqrt m) = \Order(m^{3/2})$.
\end{description}

It remains to lift this proof to the general case $k > 2$. We brute-force over all $x$-variables except for $x_i$ and $x_j$. This amounts for a factor $\Order(m^{k-2})$ in the running time. What remains is a $\MaxSP_2$ formula in the shape as before which can be solved exactly in time $\Order(m^{3/2})$ by the previous case analysis. In total this takes time $\Order(m^{k-2} m^{3/2}) = \Order(m^{k-1/2})$. The proof works in exactly the same way for minimization problems.
\end{proof}

\begin{proof}[Proof of~\autoref{lem:hyperedges}]
Let $\psi = \max_{x_1, \dots, x_k} \counting_y \phi(x_1, \dots, x_k, y)$ be a $\MaxSP_k$ formula possibly containing some hyperpredicates. We introduce a new binary relation $N(x_i, x_j)$ defined as follows: For any $x_i, x_j \in V$ it holds that $N(x_i, x_j) = \true$ if and only if $x_i$ and $x_j$ are connected by some (hyper-)edge. Observe that any (hyper-)edge contributes to at most a constant number of records $N(x_i, x_j)$, so we can construct $N$ in time $\Order(m)$ and the sparsity blows up only by a constant factor. We can now rewrite $\psi$ via 
\begin{equation*}
    \psi_0 = \max_{x_1, \dots, x_k} \counting_y \Bigg(\Bigg(\bigwedge_{i \neq j} \overline N(x_i, x_j)\Bigg) \land \phi_0(x_1, \dots, x_k, y)\Bigg)
\end{equation*}
and
\begin{equation*}
    \psi_{i,j} = \max_{x_1, \dots, x_k} \counting_y \Big(N(x_i, x_j) \land \phi(x_1, \dots, x_k, y)\Big),
\end{equation*}
where $\phi_0$ is obtained from $\phi$ by replacing all occurrences of hyperpredicates by $\false$. It follows that we can express
\begin{equation*}
    \OPT = \max\{\OPT_0, \max_{i \neq j} \OPT_{i,j} \},
\end{equation*}
where $\OPT_0$ is the optimal value of $\psi_0$, and $\OPT_{i,j}$ is the optimal value of $\psi_{i,j}$. Observe that $\psi_0$ is a $\MaxSP_k$ formula not involving any hyperpredicates, so we can by assumption $c$-approximate $\OPT_0$ in time $T(m)$. Moreover, the formulas $\psi_{i,j}$ are precisely in the shape to apply~\autoref{lem:positive-cross-edge}, so we can compute $\OPT_{i,j}$ exactly in time $\Order(m^{k-1/2})$.
\end{proof}

\paragraph*{Step 2: Removing Cross Edges}
Next, the goal is to remove all binary predicates $E(x_i, x_j)$ between two $x$-variables. Let us call these predicates $E(x_i, x_j)$ \emph{cross predicates} and the associated entries $(x_i, x_j)$ \emph{cross edges}.

\begin{lemma} \label{lem:cross-edges}
Suppose that, for any $\OptSP_k$ formula $\psi$ not containing hyperpredicates and cross predicates, $\Opt(\psi)$ can be $c$-approximated in time $\Order(m^k / s(m))$. Then $\Opt(\psi)$ can be $c$-approximated in time $\Order(m^k / s(\!\sqrt[k+1] m))$ for any $\OptSP_k$ formula $\psi$ not containing hyperpredicates.
\end{lemma}

\begin{proof}
Let $\psi = \max_{x_1, \dots, x_k} \counting_y \phi(x_1, \ldots, x_k, y)$ and let $E_1, \dots, E_r$ denote the cross predicates in the given instance. We define
\begin{equation*}
    \psi_0 := \max_{x_1, \dots, x_k} \counting_y \Bigg(\Bigg(\bigwedge_{\ell, i, j} \overline E_\ell(x_i, x_j)\Bigg) \land \phi_0(x_1, \dots, x_k, y) \Bigg)
\end{equation*}
and
\begin{equation*}
    \psi_{\ell, i, j} := \max_{x_1, \dots, x_k} \counting_y \Big( E_\ell(x_i, x_j) \land \phi(x_1, \dots, x_k, y) \Big),
\end{equation*}
where $\ell \in [r]$ and $i \neq j \in [k]$ and $\phi_0$ is the propositional formula obtained from $\phi$ by substituting all predicates $E_\ell(x_i, x_j)$ by $\false$. It is easy to verify that
\begin{equation*}
    \OPT = \max\{ \OPT_0, \max_{\ell, i, j} \OPT_{\ell, i, j} \},
\end{equation*}
where $\OPT_0$ and $\OPT_{\ell, i, j}$ are the optimal values of $\Max(\psi_0)$ and $\Max(\psi_{\ell, i, j})$, respectively. Using~\autoref{lem:positive-cross-edge}, we can compute $\OPT_{\ell, i, j}$ exactly in time $\Order(m^{k-1/2})$ for all $\ell, i, j$. It remains to efficiently solve $\Max(\psi_0)$ to compute $\OPT_0$.

As described before, we can always assume that each variable ranges over a separate set: $x_i \in X_i$, $y \in Y$. We call a vertex $x_i$ \emph{heavy} if it has degree at least $\!\sqrt[k+1]m$, and \emph{light} otherwise. The first step is to eliminate all heavy vertices; there can exist at most $\Order(m / \!\sqrt[k+1]m)$ many such vertices~$x_i$. Fixing $x_i$, we can solve the remaining problem in time $\Order(m^{k-1})$ using the baseline algorithm. We keep track of the optimal solution detected in this way. This precomputation step takes time $\Order(m^k / \!\sqrt[k+1]m)$ and afterwards we can safely remove all heavy vertices.

Next, partition each set $X_i$ into several groups $X_{i,1}, \ldots, X_{i,g}$ such that the total degree of all vertices in a group is $\Order(\!\sqrt[k+1]m)$, and the number of groups is $g = \Order(m / \!\sqrt[k+1]m)$. This is implemented by greedily inserting vertices into $X_{i,j}$ until its total degree exceeds $\!\sqrt[k+1]m$. As each vertex inserted in that way is light, we can overshoot by at most $\!\sqrt[k+1]m$.

Let $\psi_1 := \max_{x_1, \dots, x_k} \counting_y \phi_0(x_1, \dots, x_k, y)$; note that $\psi_1$ equals $\psi_0$ except that it disregards the cross predicates. Therefore, by assumption we can $c$-approximate $\Max(\psi_1)$ in time $\Order(m^k / s(m))$. The algorithm continues as follows:
\begin{enumerate}
\item For all combinations $(j_1, \ldots, j_k) \in [g]^k$, compute a $c$-approximation of $\Max(\psi_1)$ on the input $X_{1, j_1}, \ldots, X_{k, j_k}$. We keep track of the $\binom k2 m n^{k-2} + 1$ combinations with largest values (breaking ties arbitrarily).
\item For any of the top-most $\binom k2 m n^{k-2} + 1$ combinations $(j_1, \ldots, j_k)$, solve $\Max(\psi_0)$ exactly on $X_{1, j_1}, \ldots, X_{k, j_k}$ using the baseline algorithm. Return the best solution detected in this step or the precomputation phase.
\end{enumerate}
We begin with the correctness of the algorithm. First, the value of any solution $(x_1, \ldots, x_k)$ in $\Max(\psi_0)$ is at least as large as its value in $\Max(\psi_1)$. In particular, the optimal solution of $\Max(\psi_0)$ has value at least $\OPT_0$ in $\Max(\psi_1)$. We next establish an upper bound on the number \emph{false positives}, that is, tuples $(x_1,\dots,x_k)$ of different value in $\Max(\psi_0)$ than in $\Max(\psi_1)$. Observe that any such false positive contains at least one edge $(x_i, x_j)$ and since there are at most~$m$ edges, at most $\binom k2$ choices of $i, j$ and at most $n^{k-2}$ choices for the remaining vertices $x_\ell$, $\ell \neq i, j$, we can indeed bound the number of false positives by $\binom k2 m n^{k-2}$. Thus, if we witness the top-most $\binom k2 m n^{k-2} + 1$ solutions of $\Max(\psi_1)$ in step 1, among these there exists at least one solution of value $\geq \OPT_0 / c$ in $\Max(\psi_0)$.

Finally, let us bound the running time of the above algorithm. Recall that removing heavy vertices accounts for $\Order(m^k / \!\sqrt[k+1]m)$ time. In step 1, the $\Max(\psi_1)$ algorithm is applied $g^k$ times on instances of size $\Order(\!\sqrt[k+1]m)$, which takes time $\Order((m / \!\sqrt[k+1]m)^k \cdot (\!\sqrt[k+1]m)^k / s(\!\sqrt[k+1]m)) = \Order(m^k / s(\!\sqrt[k+1]m))$. Step 2 runs the baseline algorithm $m n^{k-2} = \Order(m^{k-1})$ times on instances of size $\Order(\!\sqrt[k+1]m)$, which takes time $\Order(m^{k-1} (\!\sqrt[k+1]m)^k) = \Order(m^k / \!\sqrt[k+1]m)$. Thus, the total running time is $\Order(m^k / \!\sqrt[k+1]m + m^k / s(\!\sqrt[k+1]m))$. As $s(m) \leq m$, this is as claimed. The proof for the maximization variant is complete and there are only minor adaptions necessary for minimization.
\end{proof}

\paragraph*{Step 3: Removing Parallel Edges}
After applying the previous steps we can assume that $\psi$ is an $\OptSP_k$ formula not containing hyperedges or cross edges. Let $E_1, \dots, E_r$ be the binary relations featured in $\psi$. We say that $\psi$ does not have \emph{parallel edges} if $r = 1$. In an instance of $\Opt(\psi)$ with parallel edges, any pair of vertices $(x_i, y)$ may be connected by up to $r$ parallel edges, or equivalently by an edge of $2^r$ possible \emph{colors}. We adopt the second viewpoint for this step: Let $\chi(x_i, y) := (E_1(x_i, y), \dots, E_r(x_i, y)) \in \{0, 1\}^r$ be the \emph{color} of the edge $(x_i, y)$ and let $\chi(x_1, \dots, x_k, y) := (\chi(x_1, y), \dots, \chi(x_k, y)) \in (\{0, 1\}^r)^k$ be the color of the tuple $(x_1, \dots, x_k, y)$.

\begin{lemma} \label{lem:colored-edges}
Suppose that, for any $\OptSP_k$ formula $\psi$ not containing hyperedges, cross edges and parallel edges, $\Opt(\psi)$ can be $c$-approximated in time $\Order(m^k / s(m))$. Then $\Opt(\psi)$ can be $c$-approximated in time $\Order(m^k / s(m))$ for any $\OptSP_k$ formula $\psi$ not containing hyperedges and cross edges.
\end{lemma}
\begin{proof}
Let $E_1, \dots, E_r$ denote the binary relations featured in the given instance; our goal is to construct a new instance with only a single edge predicate $E$. We leave the vertex sets~$X_i$ unchanged and construct $Y' = \{ y_\alpha : y \in Y, \alpha \in (\{0, 1\}^r)^k \}$, i.e., each vertex $y \in Y$ is copied $2^{rk} = \Order(1)$ times and each copy $y_\alpha$ is indexed by a $k$-tuple of colors $\alpha = (\alpha_1, \dots, \alpha_k) \in (\{0, 1\}^r)^k$. For every $\alpha$ we also introduce a new unary predicate~$C_\alpha$ and assign $C_\alpha(y_{\alpha'})$ if and only if $\alpha = \alpha'$.

Now let $i \in [k]$ and let $x_i \in X_i$ and $y \in Y$ be arbitrary vertices in the original instance. We assign the edges in the constructed instance as follows. If $\chi(x_i, y) = 0 = (0, \dots, 0)$, then~$x_i$ and~$y$ are not connected and we do not introduce new edges. So suppose that $\chi(x_i, y) \neq 0$. Then we add $2 \mult 2^{r(k-1)}$ edges
\begin{itemize}
\item $E(x_i, y_\beta)$, for all $\beta \in (\{0, 1\}^r)^k$ with $\beta_i = \chi(x_i, y)$, and
\item $E(x_i, y_\gamma)$, for all $\gamma \in (\{0, 1\}^r)^k$ with $\gamma_i = 0$.
\end{itemize}
Clearly the sparsity of the new instance is bounded by $2 \mult 2^{r(k-1)} m = \Order(m)$ plus the contribution of the new unary predicates which is also $\Order(m)$.

Now let $\psi = \opt_{x_1, \dots, x_k} \counting_y \phi(x_1, \dots, x_k)$. To define an equivalent $\MaxSP_k$ formula $\psi'$, for any $\alpha \in (\{0, 1\}^r)^k$ let $\phi_\alpha$ denote the formula obtained from $\phi$ by substituting $E_j(x_i, y)$ by $\true$ if $\alpha_{i, j} = 1$ and by $\false$ otherwise. We define $\psi' = \max_{x_1, \dots, x_k} \counting_y \phi'(x_1, \dots, x_k, y)$,
where $\phi'(x_1, \dots, x_k, y)$ is
\begin{equation*}
    \bigvee_{\alpha \in (\{0, 1\}^r)^k} \Bigg( \underbrace{C_\alpha(y)}_{\text{(i)}} \land \underbrace{\Bigg(\bigwedge_{i \in [k]} (E(x_i, y) \iff \alpha_i \neq 0) \Bigg)}_{\text{(ii)}} \land \underbrace{\phi_\alpha(x_1, \dots, x_k, y)}_{\text{(iii)}} \Bigg)\!.
\end{equation*}
As desired, the constructed instance contains only a single binary predicate and no cross or hyperedges. It remains to argue that the value of any tuple $(x_1, \dots, x_k)$ is not changed by the reduction. Indeed, for all $y \in Y$ we prove the following two conditions and thereby the claim.
\begin{itemize}
\item $\phi'(x_1, \dots, x_k, y_\alpha) = \phi(x_1, \dots, x_k, y)$ for $\alpha = \chi(x_1, \dots, x_k, y)$,
\item $\phi'(x_1, \dots, x_k, y_\alpha) = \false$ for all $\alpha \neq \chi(x_1, \dots, x_k, y)$.
\end{itemize}
The first bullet is simple to verify: In the evaluation of $\phi'(x_1, \dots, x_k, y_\alpha)$ we only have to focus on the $\alpha$-disjunct by the constraint (i). The constraint (ii) is satisfied by our construction of $E$ and therefore only (iii) is decisive: $\phi'(x_1, \dots, x_k, y_\alpha) = \phi_\alpha(x_1, \dots, x_k, y_\alpha) = \phi(x_1, \dots, x_k, y)$. Next, focus on the second bullet. For $\alpha \neq \chi(x_1, \dots, x_k, y)$ there exists some index $i$ such that $\alpha_i \neq \chi(x_i, y)$. By (i), we again only need to consider the $\alpha$-disjunct. We now prove that $E(x_i, y) \iff \alpha_i = 0$ which falsifies (ii) and shows $\phi'(x_1, \dots, x_k, y_\alpha) = \false$. On the one hand, if $\alpha_i \neq 0$ then there is no edge $E(x_i, y_\alpha)$, since $0 \neq \alpha_i \neq \chi(x_i, y)$. On the other hand, if $\alpha_i = 0$ then we added an edge $E(x_i, y_\alpha)$.
\end{proof}

\paragraph*{Step 4: Removing Unary Predicates}
As the final simplification, we eliminate unary predicates and show that the resulting problem can be reduced to the Hybrid Problem.

\begin{proof}[Proof of \autoref{lem:optsp-to-hybrid}]
By applying the reductions in Lemmas~\ref{lem:hyperedges},~\ref{lem:cross-edges} and~\ref{lem:colored-edges}, it suffices to show that any $\OptSP_k$ property $\psi$ not containing hyperpredicates, cross edge predicates and parallel edge predicates can be reduced to the Hybrid Problem. The shape of $\psi$ is significantly restricted and contains only the following three types of relations: Unary predicates on $X_1, \dots, X_k$, unary predicates on $Y$ and  binary predicates of the form $E(x_i, y)$ for $i \in [k]$.

We can assume that there are no unary predicates on $X_1, \dots, X_k$ as follows: By enumerating all possible assignments of these unary predicates, and by restricting the sets $X_1, \dots, X_k$ to those vertices matching the current assignment, we create a constant number of instances each without unary predicates on $X_1, \dots, X_k$.

This leaves only unary predicates on $Y$ and the edge predicates $E(x_i, y)$. Let $\psi = \opt_{x_1, \dots, x_k} \counting_y \phi(x_1, \dots, x_k, y)$. Another way to view this problem is associate a Boolean function $\phi_y : \{0, 1\}^k \to \{0, 1\}$ to every vertex $y \in Y$, which takes as input $E(x_i, y)$ and does no longer depend on the unary predicates of $y$. In that way, we can rewrite the objective as
\begin{equation*}
    \opt_{x_1, \dots, x_k} \counting_y \phi_y(E(x_1, y), \dots, E(x_k, y)).
\end{equation*}
Our goal is now to reinterpret this problem as an instance of the Hybrid Problem. As the universe, we assign
\begin{equation*}
    U = \Big\{ (y, \tau) : y \in Y, \text{$\tau \in \{0, 1\}^k$ is a satisfying assignment of $\phi_y$}\Big\},
\end{equation*}
along with the partition $U = \bigcup_{\tau \in \{0, 1\}^k} U_\tau$, $U_\tau = U \cap (Y \times \{\tau\})$. For every vertex $x_i \in X_i$, we construct a set $S_i \in \mathcal S_i$ as $S_i = \{ (y, \tau) : E(x_i, y) \} \cap U$. It is easy to check that the value of every solution is preserved in this way: $\Val(S_1, \dots, S_k) = \Val(x_1, \dots, x_k)$. The overhead of this rewriting step is $\Order(m)$ and thus negligible in the running time bound. 
\end{proof}

\bibliographystyle{plainurl}
\bibliography{refs}

\appendix
\section{Baseline Algorithm} \label{sec:baseline}
There is a simple baseline algorithm solving any $\OptSP_{k,\ell}$ problem $\Opt(\psi)$ in time $\Order(m^{k+\ell-1})$. It is a straightforward adaption of the baseline algorithm for the decision setting in~\cite{GaoIKW18} to the optimization setting.

\begin{theorem}
Let $k \geq 0$, $\ell \geq 1$ be parameters with $k + \ell \geq 2$ and let $\psi$ be an $\OptSP_{k,\ell}$ formula. Then $\Opt(\psi)$ can be solved exactly in time $\Order(m^{k+\ell-1})$. In fact, we can compute the values $\Val(x_1, \dots, x_k)$ for all tuples $(x_1, \dots, x_k)$ in the same running time.
\end{theorem} 

\begin{proof}
The proof is by induction on $k + \ell$. For $k + \ell = 2$, we obtain a linear-time improvement, and for larger $k + \ell$ we apply a downward self-reduction. We start with the improvement for $k + \ell = 2$. If $\psi \in \OptSP_{k,\ell}$, then either $\psi = \max_x \counting_y \phi(x, y)$ or $\psi = \counting_x \counting_y \phi(x, y)$. We treat both cases in a unified by viewing $x$ as a free variable and setting $\psi(x) = \counting_y \phi(x, y)$; we show that we can compute all values $\psi(x)$ for $x \in X$ in time $\Order(m)$. It follows that $\Opt(\psi)$ can be solved in time $\Order(m)$.

There are only three types of predicates: Unary predicates on $X$, unary predicates on~$Y$ and binary predicates on $X \times Y$. In each case, we can alternatively think of a single multi-colored predicate. More specifically, suppose there are $r$ unary relations $R_1, \ldots, R_r$ on~$X$. Then each object $x \in X$ is characterized by a \emph{color $\chi(x) \in \{0, 1\}^r$}, where $\chi(x)_i = 1$ if and only if $R_i(x)$ is true. We similarly define colors $\chi(y)$ and $\chi(x, y)$ for predicates on $Y$ and $X \times Y$. Moreover, we define $Y_\alpha = \{ y \in Y : \chi(y) = \alpha \}$. With this notation, we can rewrite
\begin{equation*}
  \psi(x) = \counting_{y \in Y} \phi(\chi(x), \chi(y), \chi(x, y)) = \sum_{\alpha, \beta} \phi(\chi(x), \alpha, \beta) \mult \#\{ y \in Y_\alpha : \chi(x, y) = \beta \},
\end{equation*}
where we abuse notation and interpret $\phi$ as a function which accepts colors as inputs. Let $C(x; \alpha, \beta) := \#\{ y \in Y_\alpha : \chi(x, y) = \beta \}$. Suppose we have precomputed $C(x; \alpha, \beta)$ for all colors~$\alpha, \beta$ and for all $x \in X$. Then we can compute $\psi(x)$ for all $x \in X$ in total time $\Order(|X|) = \Order(m)$ by explicitly evaluating the above expression.

It remains to show that we can precompute all values $C(x; \alpha, \beta)$ in time $\Order(m)$. There are two cases: On the one hand, if $\beta \neq 0 = (0, \dots, 0)$, then we only have to consider pairs $(x, y)$ which occur positively in at least one binary predicate. Thus, there can be at most $\Order(m)$ such pairs and a single pass through the structure suffices to compute $C(x; \alpha, \beta)$. On the other hand, if $\beta = 0$, we compute $C(x; \alpha, 0) = |Y_\alpha| - \sum_{\beta \neq 0} C(x; \alpha, \beta)$.

Next, focus on the general case $k + \ell > 2$. We evaluate
\begin{equation*}
  \psi(x_1, \dots, x_k, y_1, \dots, y_{\ell-1}) = \counting_{y_l} \phi(x_1, \dots, x_k, y_1, \dots, y_{\ell}),
\end{equation*}
for all $x_1, \dots, x_k, y_1, \dots, y_{\ell-1}$. Our goal is to brute-force over all objects $x_1$ (or analogously all objects $y_1$ in case that $k = 0$). Fixing such an object $x_1$, we are left to optimize
\begin{equation*}
  \psi'(x_2, \dots, x_k, y_1, \dots, y_{\ell-1}) = \counting_{y_{\ell}} \phi'(x_2, \dots, x_k, y_1, \dots, y_{\ell}),
\end{equation*}
for some propositional formula $\phi'$. We describe how to obtain $\phi'$ from $\phi$. Any predicate not including~$x_1$ remains untouched. Any predicate $R(x_1, x_{i_2}, \dots, x_{i_a})$ is replaced as follows: If~$a = 1$ (that is, $R$ is unary), then $R(x_1)$ is replaced by the appropriate constant depending on whether $R(x_1)$ holds for the fixed object $x_1$. If $a > 1$, then we replace $R(x_1, x_{i_2}, \dots, x_{i_a})$ by $R'(x_{i_2}, \dots, x_{i_a})$ for a new relation $R'$ of arity $a - 1$. In the reduced instance, every $R$-record $(x_1, x_{i_2}, \dots, x_{i_a})$ is replaced by an $R'$-record $(x_{i_2}, \dots, x_{i_a})$ and all original entries not containing the fixed object $x_1$ are discarded. Notice that the total number of records can only decrease in this step. It is easy to verify that the above steps are correct, and that the objective does not change.

We inductively assume that $\psi'$ can be evaluated exactly in time $\Order(m^{k+\ell-2})$. There are $n$ objects $x_1 \in X_1$, and constructing the associated instances $\psi'$ takes time $\Order(m)$. The total running time is $\Order(n m^{k+\ell-2}) = \Order(m^{k+\ell-1})$.
\end{proof}
\section{Improved Algorithms for Two or More Counting Quantifiers}\label{sec:improved-alg-multiple-counting}

\begin{theorem} \label{thm:multiple-counting}
Let $k \geq 0$, $\ell \geq 2$ be parameters with $k + \ell \geq 3$ and let $\psi$ be an $\OptSP_{k,\ell}$ formula. Then $\Opt(\psi)$ can be exactly solved in time $\Order(m^{k+\ell-3/2})$.
\end{theorem}

As a first step towards proving~\autoref{thm:multiple-counting} we will brute-force the first $k+\ell-3$ quantifiers. What remains is a $3$-quantifier problem of the form $\max_x \counting_{y, z} \phi(x, y, z)$ or $\counting_x \counting_{y, z} \phi(x, y, z)$. Both cases can be dealt with in a unified way: Let $x$ be a free variable and define $\psi(x) = \counting_{y, z} \phi(x, y, z)$. In the following we show how to evaluate $\psi(x)$ for all vertices $x$ (i.e., compute a list of $n$ values $\psi(x)$) in time $\Order(m^{3/2})$. Given these values, we can solve the original $\OptSP_{k,\ell}$ problem in time $\Order(m^{k+\ell-3} m^{3/2}) = \Order(m^{k+\ell-3/2})$.

To deal with this reduced problem, we proceed in two steps: Similar to the main reduction, we first identify a subproblem which captures the core hardness and we show how to solve this subproblem in time $\Order(m^{3/2})$ (\autoref{lem:triangle-counting}). Afterwards, we show how to derive~\autoref{thm:multiple-counting} by dealing with hyperpredicates, unary predicates and parallel predicates. The analogous statement for the model-checking case is proved in~\cite[Section 9.2]{GaoIKW18}.

\begin{lemma} \label{lem:triangle-counting}
Let $\phi : \{0, 1\}^3 \to \{0, 1\}$ be arbitrary. Given a tripartite structure $(X, Y, Z, E)$, in time $\Order(m^{3/2})$ we can evaluate $\psi(x) = \counting_{y \in Y, z \in Z} \phi(E(x, y), E(x, z), E(y, z))$ for all $x \in X$.
\end{lemma}
\begin{proof}
Let us start with the special case $\phi(a_1, a_2, a_3) = a_1 \land a_2 \land a_3$, i.e., the goal is to count the number of all triangles involving $x$, for every $x \in X$.  It is well-known how to solve this triangle counting problem in time $\Order(m^{3/2})$~\cite{AlonYZ97}; we present the algorithm here for completeness.

We call a vertex \emph{heavy} if its degree exceeds $\sqrt m$ and \emph{light} otherwise. In the first step we explicitly list all triangles involving light vertices: Enumerate all edges $(x, y) \in X \times Y$, and if~$y$ is light then further enumerate all edges $(y, z) \in Y \times Z$. For any such edge we can test in constant time whether the remaining edge $(x, z)$ is present. Afterwards, we can safely remove all light vertices in $Y$. This step takes time $\Order(m \sqrt m) = \Order(m^{3/2})$. By a analogous arguments we remove all light vertices in $X$ and $Z$. Since all remaining vertices are heavy, the graph now contains at most $\Order(m / \sqrt m) = \Order(\sqrt m)$ vertices and we can list all triangles $(x, y, z)$ in time $\Order((\sqrt m)^3) = \Order(m^{3/2})$. 

It turns out that we can reduce every function $\phi(a_1, a_2, a_3)$ to the previous case. For a set $S \subseteq [3]$, let $\phi_S(a_1, a_2, a_3) = \bigwedge_{i \in S} a_i$ and $\psi_S(x) = \counting_{y \in Y, z \in Z} \phi_S(E(x, y), E(x, z), E(y, z))$. We claim that we can compute $\psi_S(x)$ for all $x \in X$ and all $S \subseteq [3]$ in time $\Order(m^{3/2})$: In the previous paragraph we gave an algorithm to compute $\psi_{[3]}(x)$ in time $\Order(m^{3/2})$, and it is easy to see how to compute $\psi_S(x)$ in time $\Order(m)$ for all sets $S \subsetneq [3]$. We now exploit that $\{ \phi_S : S \subseteq [3] \}$ forms a basis of all Boolean functions, so we can be express
\begin{equation*}
    \psi(x) = \sum_{S \subseteq [3]} \alpha_S \psi_S(x),
\end{equation*}
for some (integer) coefficients $\alpha_S$. Therefore, having precomputed all values $\psi_S(x)$ in time $\Order(m^{3/2})$, we can compute $\psi(x)$ in time $\Order(|X|) = \Order(m)$.
\end{proof}

\begin{proof}[Proof of~\autoref{thm:multiple-counting}]
Let $\psi(x) = \counting_{y, z} \phi(x, y, z)$, where $x, y, z$ range over $X, Y, Z$, respectively. The reduction to~\autoref{lem:triangle-counting} proceeds in three steps:
\begin{description}
\item[Step 1: Removing hyperpredicates.] Let $\phi_0(x)$ denote the formula obtained from $\phi$ after substituting all occurrences of ternary predicates by $\false$, and let $\psi_0(x) = \counting_{y, z} \phi_0(x, y, z)$. We can compute the differences $\Delta(x) = \psi_0(x) - \psi(x)$ for all $x \in X$ in time $\Order(m)$ by enumerating over all hyperedges in the instance; any tuple $(x, y, z)$ not connected by a hyperedge fulfils $\phi(x, y, z) = \phi_0(x, y, z)$ and therefore does not contribute to $\Delta(x)$. In the following it suffices to compute $\psi_0(x)$ for all $x$, as we can compute $\psi(x) = \psi_0(x) + \Delta(x)$ in time $\Order(m)$.
\item[Step 2: Removing unary predicates.] By enumerating over all possible evaluations of the unary predicates in the instance, we may always restrict the sets $X$, $Y$ and $Z$ to those vertices matching the current evaluation. We may therefore assume that $\psi(x)$ does not contain unary predicates.
\item[Step 3: Removing parallel predicates.] Let $E_1, \dots, E_r$ denote the binary predicates in $\psi$. For vertices $x, y, z$, let $\chi(x, y) = (E_1(x, y), \dots, E_r(x, y)) \in \{0, 1\}^r$ denote the \emph{color} of the edge $(x, y)$ and let $\chi(x, y, z) = (\chi(x, y), \chi(x, z), \chi(y, z)) \in (\{0, 1\}^r)^3$ denote the color of the tuple $(x, y, z)$. We say that a color $\alpha \in (\{0, 1\}^r)^3$ is \emph{satisfying} if $\phi$ evaluates to true after substituting the binary predicates as specified by $\alpha$. We can assume that there is only a single satisfying color $\alpha$, as otherwise the following algorithm is simply repeated for every satisfying color (there are at at most $\Order(1)$ colors). The reduction to \autoref{lem:triangle-counting} is now almost immediate: If $\alpha_1 \neq 0 = (0, \dots, 0)$, then we can remove all edges $(x, y)$ of color different than $\alpha_1$. If $\alpha_1 = 0$, then we convert all edges $(x, y)$ of color $\chi(x, y) \neq 0$ into edges of some non-zero color. After this step, every edge $(x, y)$ has one of the two colors, so we may equivalently introduce a new binary predicate $E$ and assign $E(x, y)$ if and only if $\chi(x, y) \neq 0$. We similarly proceed for $\alpha_2$ and $\alpha_3$ and edges $(x, z)$ and $(y, z)$, respectively. The remaining problem is of the form $\counting_{y, z} \phi(E(x, y), E(x, z), E(y, z))$ for some (uniquely satisfiable) function $\phi$. Finally, we can apply \autoref{lem:triangle-counting} to solve the instance in time $\Order(m^{3/2})$. \qedhere
\end{description}
\end{proof}

\end{document}